\documentclass[11pt]{article}
\usepackage[letterpaper,hmargin=1in,vmargin=1in]{geometry}
\usepackage{microtype}
\usepackage{graphicx}                   
\usepackage{amsmath}                    
\usepackage{amssymb}                    
\usepackage{amsfonts}                   
\usepackage{url}                        
\usepackage{color}                      
\usepackage{cite}                       
\usepackage{hyperref}                   
\hypersetup{colorlinks=true,linkcolor=[rgb]{0.75,0,0},citecolor=[rgb]{0,0,0.75}}
\usepackage{amsthm}                     
\usepackage{thmtools}                    
\usepackage{thm-restate}

\declaretheorem[name=Theorem,numberwithin=section]{theorem}
\declaretheorem[name=Lemma,numberwithin=section]{lemma}

\newcommand{\ang}[1]{\langle #1\rangle}
\newcommand{\RE}{\mathbb{R}}            
\newcommand{\XX}{\mathcal{X}}           
\newcommand{\eps}{\varepsilon}          

\newcommand{\inv}[1]{\frac{1}{#1}}

\newcommand{\Gradient}{\nabla}
\newcommand{\Hess}{\nabla^2}
\newcommand{\Transpose}{\intercal}
\DeclareMathOperator{\diam}{diam}

\DeclareMathOperator{\radius}{radius}

\DeclareMathOperator{\dist}{dist}

\newcommand{\etal}{\textit{et al.}}

\DeclareSymbolFont{YHlargesymbols}{OMX}{yhex}{m}{n}
\DeclareMathAccent{\conc}{\mathord}{YHlargesymbols}{"F3}

\begin{document}

\title{Approximate Nearest Neighbor Searching with\\ Non-Euclidean and Weighted Distances\thanks{A preliminary version of this paper appeared in Proc.\ 30th Annu.\ ACM-SIAM Sympos.\ Discrete Algorithms, 2019, 355--372.}}

\author{%
    Ahmed Abdelkader\thanks{Research supported by NSF grant CCF--1618866.}\\
		Department of Computer Science \\
		University of Maryland,
		College Park, Maryland 20742 \\
		akader@cs.umd.edu
		\and
	Sunil Arya\thanks{Research supported by the Research Grants Council of Hong Kong, China under project number 16213219.}\\
		Department of Computer Science and Engineering \\
		The Hong Kong University of Science and Technology,
		Hong Kong\\
		arya@cse.ust.hk \\
		\and
	Guilherme D. da Fonseca\thanks{Research supported by the European Research Council under ERC Grant Agreement number 339025 GUDHI (Algorithmic Foundations of Geometric Understanding in Higher Dimensions).}\\
		Universit\'{e} Clermont Auvergne, LIMOS, and INRIA Sophia Antipolis,
		France\\
		fonseca@isima.fr
		\and
	David M. Mount\footnotemark[1]\\
		Department of Computer Science and 
		Institute for Advanced Computer Studies \\
		University of Maryland,
		College Park, Maryland 20742 \\
		mount@umd.edu \\
}
\date{}

\maketitle

\begin{abstract}
We present a new approach to approximate nearest-neighbor queries in fixed dimension under a variety of non-Euclidean distances. We are given a set $S$ of $n$ points in $\mathbb{R}^d$, an approximation parameter $\varepsilon > 0$, and a distance function that satisfies certain smoothness and growth-rate assumptions. The objective is to preprocess $S$ into a data structure so that for any query point $q$ in $\mathbb{R}^d$, it is possible to efficiently report any point of $S$ whose distance from $q$ is within a factor of $1+\varepsilon$ of the actual closest point.

Prior to this work, the most efficient data structures for approximate nearest-neighbor searching in spaces of constant dimensionality applied only to the Euclidean metric. This paper overcomes this limitation through a method called convexification. For admissible distance functions, the proposed data structures answer queries in logarithmic time using $O(n \log (1 / \varepsilon) / \varepsilon^{d/2})$ space, nearly matching the best known bounds for the Euclidean metric. These results apply to both convex scaling distance functions (including the Mahalanobis distance and weighted Minkowski metrics) and Bregman divergences (including the Kullback-Leibler divergence and the Itakura-Saito distance).
\end{abstract}

\section{Introduction} \label{s:intro}

Nearest-neighbor searching is a fundamental retrieval problem with numerous applications in fields such as machine learning, data mining, data compression, and pattern recognition. A set of $n$ points, called \emph{sites}, is preprocessed into a data structure such that, given any query point $q$, it is possible to report the site that is closest to $q$. The most common formulation involves points in real $d$-dimensional space, $\RE^d$, under the Euclidean metric. Unfortunately, the best solution achieving $O(\log n)$ query time uses roughly $O(n^{d/2})$ storage space~\cite{Cla88}, which is too high for many applications.

This has motivated the study of approximations. Given an approximation parameter $\eps > 0$, \emph{$\eps$-approximate nearest-neighbor searching} ($\eps$-ANN) returns any site whose distance from $q$ is within a factor of $1+\eps$ of the distance to the true nearest neighbor. Throughout, we focus on $\RE^d$ for fixed $d$ and on data structures that achieve query time $O(\log\frac{n}{\eps})$. The objective is to produce data structures of linear storage while minimizing the dependencies on $\eps$, which typically grow exponentially with the dimension. 

Har-Peled showed that logarithmic query time could be achieved for Euclidean $\eps$-ANN queries using roughly $O(n/\eps^{d})$ space through the \emph{approximate Voronoi diagram} (AVD) data structure~\cite{Har01}. Despite subsequent work on the problem (see, e.g., \cite{AMM09a,AFM18a}), the storage requirements needed to achieve logarithmic query time remained essentially unchanged for over 15 years. Recently, Arya {\etal}~\cite{AFM17a,AFM17b} succeeded in reducing the storage to $O(n / \eps^{d/2})$ by applying techniques from convex approximation.%
\footnote{Chan~\cite{Cha18} presented a similar result by a very different approach, and it generalizes to some other distance functions, however the query time is not logarithmic.}
Unlike the simpler data structure of \cite{Har01}, which can be applied to a variety of metrics, this recent data structure exploits properties that are specific to Euclidean space, which significantly limits its applicability. In particular, it applies a reduction to approximate polytope membership~\cite{AFM18a} based on the well-known \emph{lifting transformation}~\cite{textbook}. However, this transformation applies only for the Euclidean distance. Note that all the aforementioned data structures rely on the triangle inequality. Therefore, they cannot generally be applied to situations where each site is associated with its own distance function as arises, for example, with multiplicatively weighted sites (defined below).

Har-Peled and Kumar introduced a powerful technique to overcome this limitation through the use of \emph{minimization diagrams} \cite{HaK15}. For each site $p_i$, let $f_i : \RE^d \rightarrow \RE^+$ be the associated distance function. Let $\mathcal{F}_{\min}$ denote the pointwise minimum of these functions, that is, the \emph{lower-envelope function}. Clearly, approximating the value of $\mathcal{F}_{\min}$ at a query point $q$ is equivalent to approximating the distance to $q$'s nearest neighbor.%
\footnote{The idea of using envelopes of functions for the purpose of nearest-neighbor searching has a long history, and it is central to the well-known relationship between the Euclidean Voronoi diagram of a set of points in $\RE^d$ and the lower envelope of a collection of hyperplanes in $\RE^{d+1}$ through the lifting transformation \cite{textbook}.}
Har-Peled and Kumar proved that $\eps$-ANN searching over a wide variety of distance functions (including additively and multiplicatively weighted sites) could be cast in this manner \cite{HaK15}. They formulated this problem in a very abstract setting, where no explicit reference is made to sites. Instead, the input is expressed in terms of abstract properties of the distance functions, such as their growth rates and ``sketchability.'' While this technique is very general, the complexity bounds are much worse than for the corresponding concrete versions. For example, in the case of Euclidean distance with multiplicative weights, in order to achieve logarithmic query time, the storage used is $O((n \log^{d+2} n) / \eps^{2d+2} + n / \eps^{d^2+d})$. Similar results are achieved for a number of other distance functions that are considered in~\cite{HaK15}.

\begin{figure*}[tp]
  \centerline{\includegraphics[scale=0.60]{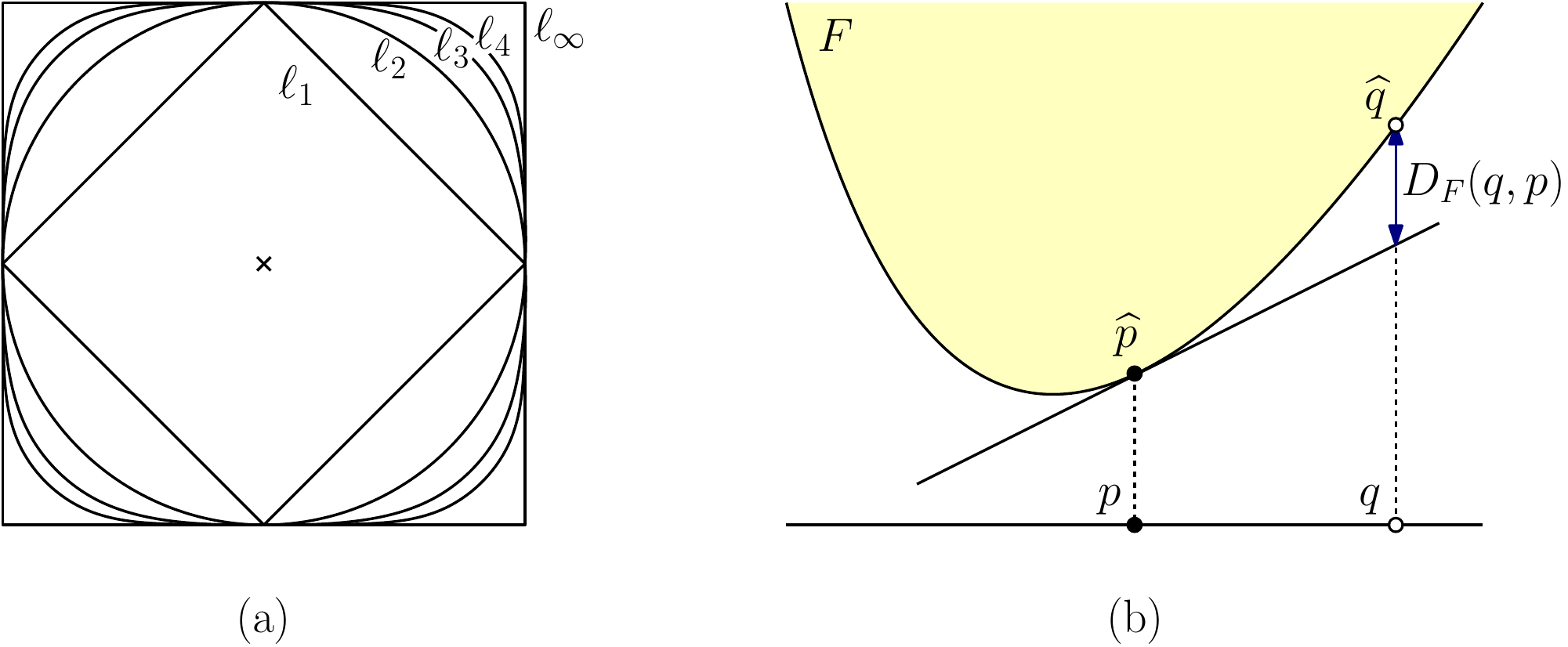}}
  \caption{(a) Unit balls in different Minkowski norms and (b) a geometric interpretation of the Bregman divergence.}
  \label{f:distances}
\end{figure*}

This motivates the question of whether it is possible to answer ANN queries for non-Euclidean distance functions while matching the best bounds for Euclidean ANN queries. In this paper, we present a general approach for designing such data structures achieving $O(\log\frac{n}{\eps})$ query time and $O((n/\eps^{d/2}) \log\inv\eps)$ storage. Thus, we suffer only an extra $\log\inv\eps$ factor in the space bounds compared to the best results for Euclidean $\eps$-ANN searching. We demonstrate the power of our approach by applying it to a number of natural problems:

\begin{description}
\item[Minkowski Distance:] The $\ell_k$ distance (see Figure~\ref{f:distances}(a)) between two points $p$ and $q$ is defined as $\|q-p\|_k = (\sum_{i=1}^d |p_i - q_i|^k)^\frac{1}{k}$. Our results apply for any real constant $k > 1$.
   
\item[Multiplicative Weights:] Each site $p$ is associated with a weight $w_p > 0$ and $f_p(q) = w_p \|q-p\|$. The generalization of the Voronoi diagram to this distance function is known as the \emph{M\"{o}bius diagram} \cite{BoK03}. Our results generalize from $\ell_2$ to any Minkowski $\ell_k$ distance, for constant $k > 1$.

\item[Mahalanobis Distance:] Each site $p$ is associated with a $d \times d$ positive-definite matrix $M_p$ and $f_p(q) = \sqrt{(p-q)^{\Transpose} M_p (p-q)}$. Mahalanobis distances are widely used in machine learning and statistics. Our results hold under the assumption that for each point $p$, the ratio between the maximum and minimum eigenvalues of $M_p$ is bounded.
    
\item[Scaling Distance Functions:] Each site $p$ is associated with a closed convex body $K_p$ whose interior contains the origin, and $f_p(q)$ is the smallest $r$ such that $(q - p)/r \in K_p$ (or zero if $q = p$). (These are also known as \emph{convex distance functions}~\cite{ChD85}.) These generalize and customize normed metric spaces by allowing metric balls that are not centrally symmetric and allowing each site to have its own distance function.
\end{description}

Scaling distance functions generalize the Minkowski distance, multiplicative weights, and the Mahalanobis distance. Our results hold under the assumption that the convex body $K_p$ inducing the distance function satisfies two assumptions. First, it needs to be \emph{fat} in the sense that it can be sandwiched between two Euclidean balls centered at the origin whose radii differ by a constant factor. Second, it needs to be \emph{smooth} in the sense that the radius of curvature for every point on $K_p$'s boundary is within a constant factor of its diameter. (Formal definitions will be given in Section~\ref{s:scaling-admiss}.)

\begin{theorem}[ANN for Scaling Distances] \label{thm:convexdf}
Given an approximation parameter $0 < \eps \le 1$ and a set $S$ of $n$ sites in $\RE^d$ where each site $p \in S$ is associated with a fat, smooth convex body $K_p \subset \RE^d$ (as defined above), there exists a data structure that can answer $\eps$-approximate nearest-neighbor queries with respect to the respective scaling distance functions defined by $K_p$ with
\[
	\hbox{Query time:~} O\kern-2pt \left(\log \frac{n}{\eps} \right)
	\quad\hbox{and}\quad
	\textrm{Space:~} O\kern-2pt \left( \frac{n \log\inv\eps}{\eps^{d/2}} \right).
\]
\end{theorem}

Another important application that we consider is the Bregman divergence. Examples of Bregman divergences include the square of any Mahalanobis distance (and hence the square of the Euclidean distance)~\cite{BNN10}, the Kullback-Leibler divergence (also known as relative entropy)~\cite{KuL51}, and the Itakura-Saito distance~\cite{ItS68} among others. They have numerous applications in machine learning and computer vision~\cite{BMD05,STG10}.

\begin{description}
\item[Bregman Divergence:] Given an open convex domain $\XX \subseteq \RE^d$, a strictly convex and differentiable real-valued function $F$ on $\XX$, and $q, p \in \XX$, the \emph{Bregman divergence} of $q$ from $p$ is
\[
    D_F(q, p)
   		~ = ~ F(q) - (F(p) +  \ang{\Gradient F(p), q-p}),
\]
where $\Gradient F$ denotes the gradient of $F$ and ``$\ang{\cdot,\cdot}$'' denotes the standard dot product.  
\end{description}

The Bregman divergence has the following geometric interpretation (see Figure~\ref{f:distances}(b)). Let $\widehat{p}$ denote the vertical projection of $p$ onto the graph of $F$, that is, $(p, F(p))$, and define $\widehat{q}$ similarly. $D_F(q, p)$ is the vertical distance between $\widehat{q}$ and the hyperplane tangent to $F$ at the point $\widehat{p}$. Equivalently, $D_F(q, p)$ is just the error that results by estimating $F(q)$ by a linear model at $p$. 

Bregman divergences lack many of the properties of typical distance functions. They are generally not symmetric, that is, $D_F(q, p) \neq D_F(p, q)$ and need not satisfy the triangle inequality. As a function of $q$ alone, Bregman divergences are convex functions. Throughout, we treat the first argument $q$ as the query point and the second argument $p$ as the site, but it is possible to reverse these through dualization \cite{BNN10}.

Data structures have been presented for answering exact nearest-neighbor queries in the Bregman divergence by Cayton~\cite{Cay08} and Nielson {\etal}~\cite{NPB09}, but no complexity analysis was given. Worst-case bounds have been achieved by imposing restrictions on the function $F$. A variety of complexity measures have been proposed, including the following. Given a parameter $\mu \ge 1$, and letting $\|p-q\|$ denote the Euclidean distance between $p$ and $q$:
\begin{itemize}
\item $D_F$ is \emph{$\mu$-asymmetric} if for all $p, q \in \XX$, $D_F(q, p) \leq \mu D_F(p, q)$.

\item $D_F$ is \emph{$\mu$-similar}%
\footnote{Our definition of $\mu$-similarity differs from that of \cite{AcB09}. First, we have replaced $1/\mu$ with $\mu$ for compatibility with asymmetry. Second, their definition allows for any Mahalanobis distance, not just Euclidean. This is a trivial distinction in the context of nearest-neighbor searching, since it is possible to transform between such distances by applying an appropriate positive-definite linear transformation to the query space.}
if for all $p, q \in \XX$, $\|q - p\|^2 \leq D_F(q, p) \leq \mu \kern+1pt \|q - p\|^2$.
\end{itemize}

Abdullah {\etal}~\cite{AMV13} presented data structures for answering $\eps$-ANN queries for decomposable%
\footnote{A Bregman divergence over $\RE^d$ is \emph{decomposable} if it can be expressed as the $d$-fold sum of a $1$-dimensional Bregman divergence applied to each coordinate.}
Bregman divergences in spaces of constant dimension under the assumption of bounded similarity. Later, Abdullah and Venkatasubramanian~\cite{AbV15} established lower bounds on the complexity of Bregman ANN searching under the assumption of bounded asymmetry.

Our results for ANN searching in the Bregman divergence are stated below. They hold under a related measure of complexity, called \emph{$\tau$-admissibility}, which is more inclusive (that is, weaker) than $\mu$-similarity, but seems to be more restrictive than $\mu$-asymmetry. It is defined in Section~\ref{s:bregman-complexity}, where we also explore the relationships between these measures.

\begin{theorem}[ANN for Bregman Divergences] \label{thm:bregman}
Given a $\tau$-admissible Bregman divergence $D_F$ for a constant $\tau$ defined over an open convex domain $\XX \subseteq \RE^d$, a set $S$ of $n$ sites in $\RE^d$, and an approximation parameter $0 < \eps \le 1$, there exists a data structure that can answer $\eps$-approximate nearest-neighbor queries with respect to $D_F$ with
\[
	\hbox{Query time:~} O\kern-2pt \left(\log \frac{n}{\eps} \right)
	\quad\hbox{and}\quad
	\textrm{Space:~} O\kern-2pt \left( \frac{n \log\inv\eps}{\eps^{d/2}} \right).
\]
The dependence on $\tau$ in the query time is $O(\log \tau)$ and $O(\tau^d)$ in the storage.
\end{theorem}

\medskip

Note that our results are focused on the \emph{existence} of these data structures, and construction time is not discussed. While we see no significant impediments to their efficient construction by modifying the constructions of related data structures, a number of technical results would need to be developed. We therefore leave the question of efficient construction as a rather technical but nonetheless important open problem.

\subsection{Methods}

Our solutions are all based on the application of a technique, called \emph{convexification}. Recently, the authors showed how to efficiently answer several approximation queries with respect to convex polytopes~\cite{AFM17a,AFM17b,AFM18b,AbM18}, including polytope membership, ray shooting, directional width, and polytope intersection. As mentioned above, the linearization technique using the lifting transformation can be used to produce convex polyhedra for the sake of answering ANN queries, but it is applicable only to the Euclidean distance (or more accurately the squared Euclidean distance and the related power distance~\cite{Aur87a}). In the context of approximation, polytopes are not required. The convex approximation methods described above can be adapted to work on any convex body, even one with curved boundaries. This provides us with an additional degree of flexibility. Rather than applying a transformation to linearize the various distance functions, we can go a bit overboard and ``convexify'' them.

Convexification techniques have been used in non-linear optimization for decades~\cite{Ber79}, for example the $\alpha$BB optimization method locally convexifies constraint functions to produce constraints that are easier to process~\cite{AMF95}. However, we are unaware of prior applications of this technique in computational geometry in the manner that we use it. (For an alternate use, see \cite{Cla06}.)

The general idea involves the following two steps. First, we apply a quadtree-like approach to partition the query space (that is, $\RE^d$) into cells so that the restriction of each distance function within each cell has certain ``nice'' properties, which make it possible to establish upper bounds on the gradients and the eigenvalues of their Hessians. We then add to each function a common ``convexifying'' function whose Hessian has sufficiently small (in fact negative) eigenvalues, so that all the functions become concave (see Figure~\ref{f:convexify} in Section~\ref{s:convexify} below). We then exploit the fact that the lower envelope of concave functions is concave. The region lying under this lower envelope can be approximated by standard techniques, such as the ray-shooting data structure of \cite{AFM17a}. We show that if the distance functions satisfy some \emph{admissibility conditions}, this can be achieved while preserving the approximation errors.

\medskip

The rest of the paper is organized as follows. In the next section we present definitions and preliminary results. Section~\ref{s:convexify} discusses the concept of convexification, and how it is applied to vertical ray shooting on the minimization diagram of sufficiently well-behaved functions. In Section~\ref{s:scaling-dist}, we present our solution to ANN searching for scaling distance functions, proving Theorem~\ref{thm:convexdf}. In Section~\ref{s:bregman-ann}, we do the same for the case of Bregman divergence, proving Theorem~\ref{thm:bregman}. Finally, in Section~\ref{s:deferred} we present technical details that have been deferred from the main body of the paper.

\section{Preliminaries} \label{s:prelim}

In this section we present a number of definitions and results that will be useful throughout the paper.

\subsection{Definitions and Notation} \label{s:notation}

For the sake of completeness, let us recall some standard definitions. Given a function $f: \RE^d \rightarrow \RE$, its \emph{graph} is the set of $(d+1)$-dimensional points $(x,f(x))$, for $x \in \RE^d$. For the sake of illustration, we treat the last coordinate axis as being directed vertically upwards. The function's \emph{epigraph} is the set of points on or above the graph, and its \emph{hypograph} is the set of points on or below the graph.

The gradient and Hessian of a function generalize the concepts of the first and second derivative to a multidimensional setting. The \emph{gradient} of $f$, denoted $\Gradient f$, is defined as the vector field $\big( \frac{\partial f}{\partial x_1}, \ldots, \frac{\partial f}{\partial x_d} \big)^{\Transpose}$. The gradient vector points in a direction in which the function grows most rapidly. For any point $x$ and any unit vector $v$, the rate of change of $f$ along $v$ is given by the dot product $\ang{\Gradient f(x), v}$. The \emph{Hessian} of $f$ at $x$, denoted $\Hess f(x)$, is a $d \times d$ matrix of second-order partial derivatives at $x$. For twice continuously differentiable functions, $\Hess f(x)$ is symmetric, implying that it has $d$ (not necessarily distinct) real eigenvalues. 

Given a $d$-vector $v$, let $\|v\|$ denote its length under the \emph{Euclidean norm}, and the \emph{Euclidean distance} between points $p$ and $q$ is $\|q - p\|$. Given a $d \times d$ matrix $A$, its \emph{spectral norm} is $\|A\| = \sup \, \{\|Ax\| \, / \, \|x\| \,:\, x \in \RE^d \mbox{ and } x \neq 0\}$. Since the Hessian is a symmetric matrix, it follows that $\|\Hess f(x)\|$ is the largest absolute value attained by the eigenvalues of $\Hess f(x)$.

A real-valued function $f$ defined on a nonempty subset $\XX$ of $\RE^d$ is \emph{convex} if the domain $\XX$ is convex and for any $x,y \in \XX$ and $\alpha \in [0,1]$, $f(\alpha x + (1 - \alpha) y) \leq \alpha f(x) + (1 - \alpha) f(y)$. It is \emph{concave} if $-f$ is convex. A twice continuously differentiable function on a convex domain is convex if and only if its Hessian matrix is positive semidefinite in the interior of the domain. It follows that all the eigenvalues of the Hessian of a convex function are nonnegative.
 
 Given a function $f: \RE^d \rightarrow \RE$ and a closed Euclidean ball $B$ (or generally any closed bounded region), let $f^+(B)$ and $f^-(B)$ denote the maximum and minimum values, respectively, attained by $f(x)$ for $x \in B$. Similarly, define $\|\Gradient f^+(B)\|$ and $\|\Hess f^+(B)\|$ to be the maximum values of the norms of the gradient and Hessian, respectively, for any point in $B$.

\subsection{Minimization Diagrams and Vertical Ray Shooting} \label{s:ray-shoot}

Consider a convex domain $\XX \subseteq \RE^d$ and a set of functions $\mathcal{F} = \{f_1, \ldots, f_m\}$, where $f_i : \XX \rightarrow \RE^+$. Let $\mathcal{F}_{\min}$ denote the associated \emph{lower-envelope function}, that is $\mathcal{F}_{\min}(x) = \min_{1 \leq i \leq m} f_i(x)$. As Har-Peled and Kumar \cite{HaK15} observed, for any $\eps > 0$, we can answer $\eps$-ANN queries on any set $S$ by letting $f_i$ denote the distance function to the $i$th site, and computing any index $i$ (called a \emph{witness}) such that $f_i(q) \leq (1+\eps) \mathcal{F}_{\min}(q)$.

We can pose this as a geometric approximation problem in one higher dimension. Consider the hypograph in $\RE^{d+1}$ of $\mathcal{F}_{\min}$. Answering $\eps$-ANN queries in the above sense can be thought of as approximating the result of a vertical ray shot upwards from the point $(q,0) \in \RE^{d+1}$ until it hits the lower envelope, where the allowed approximation error is $\eps \mathcal{F}_{\min}(q)$. Because the error is relative to the value of $\mathcal{F}_{\min}(q)$, this is called a \emph{relative $\eps$-AVR query}. It is also useful to consider a variant in which the error is absolute. An \emph{absolute $\eps$-AVR query} returns any witness $i$ such that $f_i(q) \le \eps + \mathcal{F}_{\min}(q)$ (see Fig.~\ref{f:ray-shoot}).

The hypograph of a general minimization diagram can be unwieldy. Our approach to answer AVR queries efficiently will involve subdividing space into regions such that within each region it is possible to transform the hypograph into a convex shape. In the next section, we will describe this transformation. Given this, our principal utility for answering $\eps$-AVR queries efficiently is encapsulated in the following lemma (see Figure~\ref{f:ray-shoot}). The proof has been deferred to Section~\ref{s:deferred}.

\begin{figure}[htbp]
  \centerline{\includegraphics[scale=0.40]{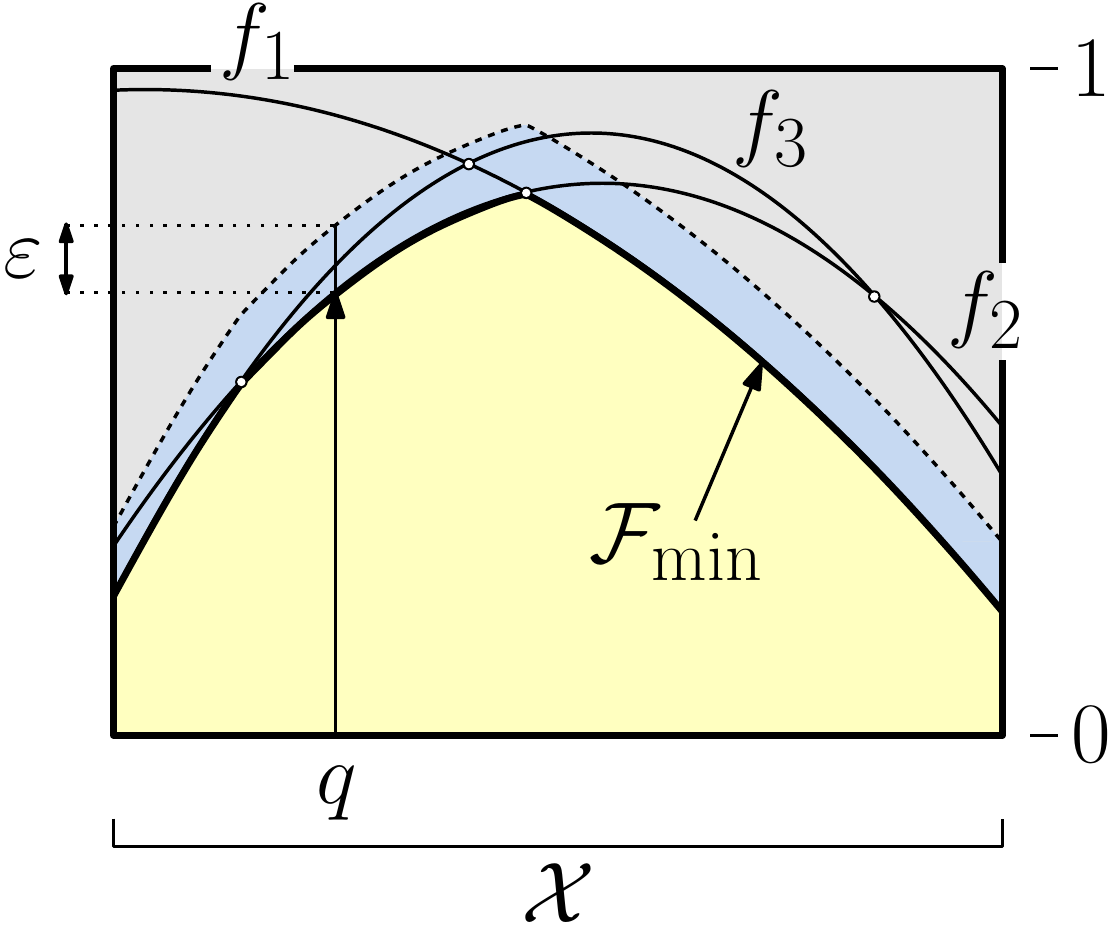}}
  \caption{Approximate AVR query assuming absolute errors. For the query $q$, the exact answer is $f_2$, but $f_3$ would be acceptable.}
  \label{f:ray-shoot}
\end{figure}

\begin{restatable}{lemma}{rayshootvertical}
\label{l:ray-shoot-vertical}
\emph{(Answering $\eps$-AVR Queries)} 
Consider a unit ball $B \subseteq \RE^d$ and a family of concave functions $\mathcal{F} = \{f_1, \ldots, f_m\}$ defined over $B$ such that for all $1 \leq i \leq m$ and $x \in B$, $f_i(x) \in [0,1]$ and $\|\Gradient f_i(x)\| \leq 1$. Then, for any $0 < \eps \leq 1$, there is a data structure that can answer absolute $\eps$-AVR queries in time $O\big( \log \inv{\eps} \big)$ and storage $O\big( (\inv{\eps})^{d/2} \big)$.
\end{restatable}

\section{Convexification} \label{s:convexify}

In this section we discuss the key technique underlying many of our results. As mentioned above, our objective is to answer $\eps$-AVR queries with respect to the minimization diagram, but this is complicated by the fact that it does not bound a convex set.

In order to overcome this issue, let us make two assumptions. First, we restrict the functions to a bounded convex domain, which for our purposes may be taken to be a closed Euclidean ball $B$ in $\RE^d$. Second, let us assume that the functions are smooth, implying in particular that each function $f_i$ has a well defined gradient $\Gradient f_i$ and Hessian $\Hess f_i$ for every point of $B$. As mentioned above, a function $f_i$ is convex (resp., concave) over $B$ if and only if all the eigenvalues of $\Hess f_i(x)$ are nonnegative (resp., nonpositive). Intuitively, if the functions $f_i$ are sufficiently well-behaved it is possible to compute upper bounds on the norms of the gradients and Hessians throughout $B$. Given $\mathcal{F}$ and $B$, let $\Lambda^+$ denote an upper bound on the largest eigenvalue of $\Hess f_i(x)$ for any function $f_i \in \mathcal{F}$ and for any point $x \in B$. 

We will apply a technique, called \emph{convexification}, from the field of nonconvex optimization \cite{Ber79,AMF95}. If we add to $f_i$ any function whose Hessian has a maximum eigenvalue at most $-\Lambda^+$, we will effectively ``overpower'' all the upward curving terms, resulting in a function having only nonpositive eigenvalues, that is, a concave function.%
\footnote{While this intuition is best understood for convex functions, it can be applied whenever there is an upper bound on the maximum eigenvalue.}
The lower envelope of concave functions is concave, and so techniques for convex approximation (such as Lemma~\ref{l:ray-shoot-vertical}) can be applied to the hypograph of the resulting lower-envelope function.

To make this more formal, let $p \in \RE^d$ and $r \in \RE_{\geq 0}$ denote the center point and radius of $B$, respectively. Define a function $\phi$ (which depends on $B$ and $\Lambda^+$) to be
\[
    \phi(x) 
        ~ = ~ \frac{\Lambda^+}{2} \left( r^2 - \sum_{j=1}^d (x_j - p_j)^2 \right)
        ~ = ~ \frac{\Lambda^+}{2} (r^2 - \|x - p\|^2).
\]
It is easy to verify that $\phi$ evaluates to zero along $B$'s boundary and is positive within $B$'s interior. Also, for any $x \in \RE^d$, the Hessian of $\|x - p\|^2$ (as a function of $x$) is a $d \times d$ diagonal matrix $2 I$, and therefore $\Hess \phi(x) = -\Lambda^+ I$. Now, for $1 \le i \le m$, define $\conc{f}_{i}(x) ~ = ~ f_i(x) + \phi(x)$ and
\[
        \conc{F}_{\min}(x) 
            ~ = \min_{1 \leq i \leq m} \conc{f}_{i}(x) 
            ~ = ~ \mathcal{F}_{\min}(x) + \phi(x).
\]

Because all the functions are subject to the same offset at each point $x$, $\conc{F}_{\min}$ preserves the relevant combinatorial structure of $\mathcal{F}_{\min}$, and in particular $f_i$ yields the minimum value to $\mathcal{F}_{\min}(x)$ at some point $x$ if and only if $\conc{f}_i$ yields the minimum value to $\conc{F}_{\min}(x)$. Absolute vertical errors are preserved as well. Observe that $\conc{F}_{\min}(x)$ matches the value of $\mathcal{F}_{\min}$ along $B$'s boundary and is larger within its interior. Also, since $\Hess \phi(x) = -\Lambda^+ I$, it follows from elementary linear algebra that each eigenvalue of $\Hess \conc{f}_i(x)$ is smaller than the corresponding eigenvalue of $\Hess f_i(x)$ by $\Lambda^+$. Thus, all the eigenvalues of $\conc{f}_i(x)$ are nonpositive, and so $\conc{f}_i$ is concave over $B$. In turn, this implies that $\conc{F}_{\min}$ is concave, as desired. We will show that, when properly applied, relative errors are nearly preserved, and hence approximating the convexified lower envelope yields an approximation to the original lower envelope.

\subsection{A Short Example} \label{s:example}

As a simple application of this technique, consider the following problem. Let $\mathcal{F} = \{f_1, \ldots, f_m\}$ be a collection of $m$ multivariate polynomial functions over $\RE^d$ each of constant degree and having coefficients whose absolute values are $O(1)$ (see Figure~\ref{f:convexify}(a)). It is known that the worst-case combinatorial complexity of the lower envelope of algebraic functions of fixed degree in $\RE^d$ lies between $\Omega(n^d)$ and $O(n^{d+\alpha})$ for any $\alpha > 0$ \cite{Sha94}, which suggests that any exact solution to computing a point on the lower envelope $\mathcal{F}_{\min}$ will either involve high space or high query time. 

\begin{figure*}[htbp]
  \centerline{\includegraphics[scale=0.40]{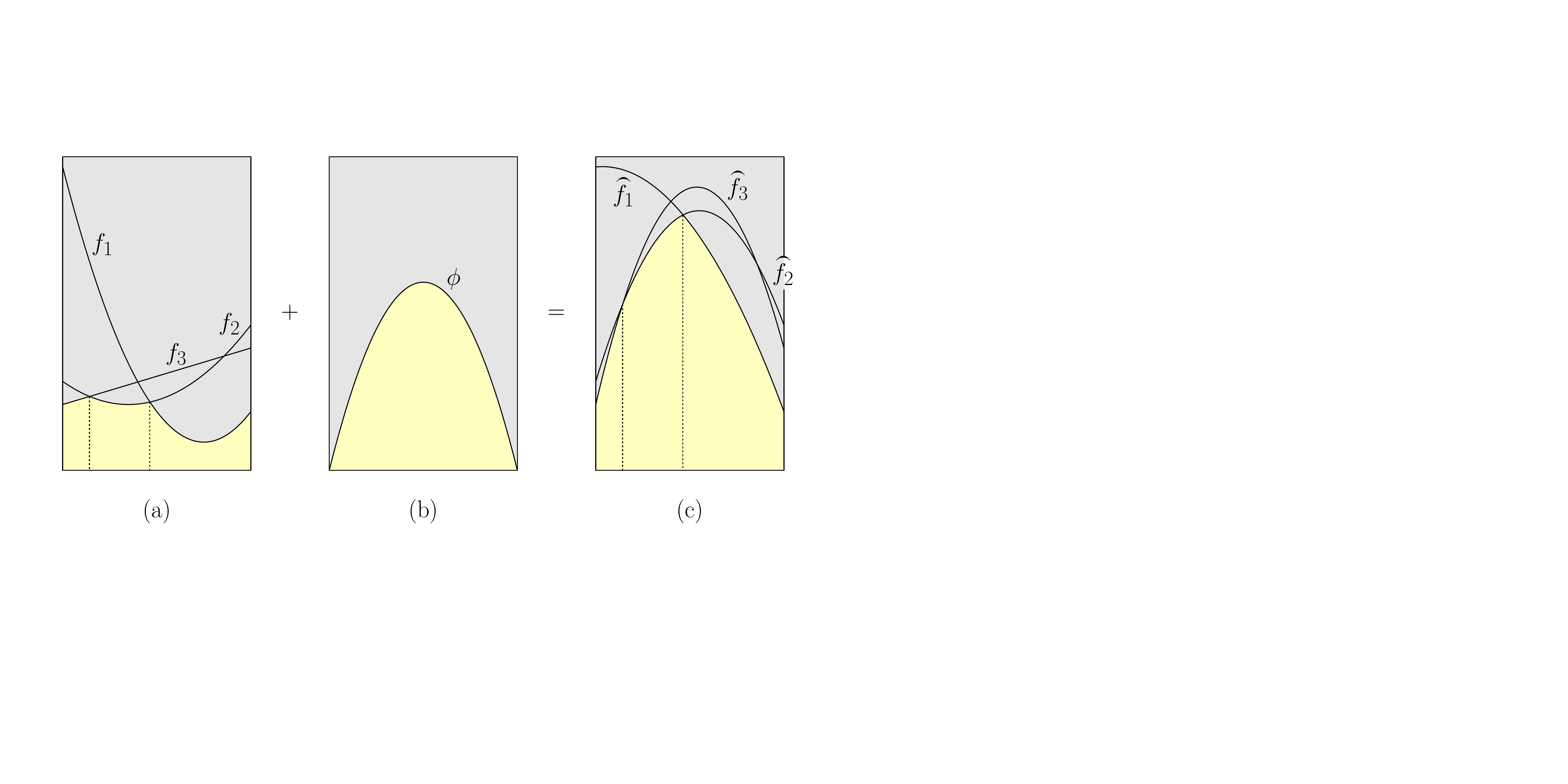}}
  \caption{Convexification. (a) The original functions $f_i$, (b) the convexifying function $\phi$, and (c) the modified functions $\conc{f}_i = f_i + \phi$.}
  \label{f:convexify}
\end{figure*}

Let us consider a simple approximate formulation by restricting $\mathcal{F}$ to a unit $d$-dimensional Euclidean ball $B$ centered at the origin. Given a parameter $\eps > 0$, the objective is to compute for any query point $q \in \RE^d$ an \emph{absolute $\eps$-approximation} by returning the index of a function $f_i$ such that $f_i(q) \le \mathcal{F}_{\min}(q) + \eps$. (While relative errors are usually desired, this simpler formulation is sufficient to illustrate how convexification works.) Since the degrees and coefficients are bounded, it follows that for each $x \in B$, the norms of the gradients and Hessians for each function $f_i$ are bounded. A simple naive solution would be to overlay $B$ with a grid with cells of diameter $\Theta(\eps)$ and compute the answer for a query point centered within each grid cell. Because the gradients are bounded, the answer to the query for the center point is an absolute $\eps$-approximation for any point in the cell. This produces a data structure with space $O((\inv{\eps})^d)$.

To produce a more space-efficient solution, we apply convexification. Because the eigenvalues of the Hessians are bounded for all $x \in B$ and all functions $f_i$, it follows that there exists an upper bound $\Lambda^+ = O(1)$ on all the Hessian eigenvalues. Therefore, by computing the convexifying function $\phi$ described above (see Figure~\ref{f:convexify}(b)) to produce the new function $\conc{F}_{\min}$ (see Figure~\ref{f:convexify}(c)) we obtain a concave function. It is easy to see that $\phi$ has bounded gradients and therefore so does $\conc{F}_{\min}$. The hypograph of the resulting function when suitably trimmed is a convex body of constant diameter residing in $\RE^{d+1}$. After a suitable scaling (which will be described later in Lemma~\ref{l:ray-shoot-general}), the functions can be transformed so that we may apply Lemma~\ref{l:ray-shoot-vertical} to answer approximate vertical ray-shooting queries in time $O(\log \inv{\eps})$ with storage $O((\inv{\eps})^{d/2})$. This \emph{halves} the exponential dependence in the dimension over the simple approach. 

\subsection{Admissible Distance Functions} \label{s:admissibility}

In this section we show that approximation errors can be bounded if the distance functions satisfy certain admissibility properties. We are given a domain $\XX \subseteq \RE^d$, and we assume that each distance function is associated with a defining site $p \in \XX$. Consider a distance function $f_p : \XX \rightarrow \RE^+$ with a well-defined gradient and Hessian for each point of $\XX$.%
\footnote{This assumption is really too strong, since distance functions often have undefined gradients or Hessians at certain locations (e.g., the sites themselves). For our purposes it suffices that the gradient and Hessian are well defined at any point within the region where convexification will be applied.}
Given $\tau > 0$, we say that $f_p$ is \emph{$\tau$-admissible} if for all $x \in \XX$:
\begin{enumerate}
\item[$(i)$] $\| \Gradient f_p(x)\| \kern+1pt \| x - p \| ~ \leq ~ \tau f_p(x)$, and

\item[$(ii)$] $\|\Hess \kern-1pt f_p(x)\| \kern+1pt \|x - p\|^2 ~ \leq ~ \tau^2 f_p(x)$.
\end{enumerate}
Intuitively, an admissible function exhibits growth rates about the site that are polynomially upper bounded. For example, it is easy to prove that $f_p(x) = \|x - p\|^c$ is $O(c)$-admissible, for any $c \geq 1$. 

Admissibility implies bounds on the magnitudes of the function values, gradients, and Hessians. Given a Euclidean ball $B$ and site $p$, we say that $B$ and $p$ are \emph{$\beta$-separated} if $\dist(p, B) / \diam(B) \geq \beta$ (where $\dist(p, B)$ is the minimum Euclidean distance between $p$ and $B$ and $\diam(B)$ is $B$'s diameter). The following lemma presents upper bounds on $f^+(B)$, $\|\Gradient f^+(B)\|$, and $\|\Hess f^+(B)\|$ in terms of these quantities. (Recall the definitions from Section~\ref{s:notation}.)

\begin{restatable}{lemma}{admissutil}
\label{l:admiss-util}
Consider an open convex domain $\XX$, a site $p \in \XX$, a $\tau$-admissible distance function $f_p$, and a Euclidean ball $B \subset \XX$. If $B$ and $p$ are $(\tau \kappa)$-separated for $\kappa > 1$, then:
\begin{enumerate}
\item[$(i)$] $f_p^+(B) ~\leq~ f_p^-(B) \kappa / (\kappa - 1)$,

\item[$(ii)$] $\|\Gradient f_p^+(B)\| ~\leq~ f_p^+(B) / (\kappa \cdot \diam(B))$, and

\item[$(iii)$] $\|\Hess f_p^+(B)\| ~\leq~ f_p^+(B) / (\kappa \cdot \diam(B))^2$.
\end{enumerate}
\end{restatable}

\begin{proof} 
To prove (i), let $x^+$ and $x^-$ denote the points of $B$ that realize the values of $f_p^+(B)$ and $f_p^-(B)$, respectively. By applying the mean value theorem, there exists a point $s$ on the line segment $\overline{x^- x^+}$ such that $f_p^+(B) - f_p^-(B) = \ang{\Gradient f_p(s), x^+ - x^-}$. By the Cauchy-Schwarz inequality
\[
    f_p^+(B) - f_p^-(B)
        ~  =   ~ \ang{\Gradient f_p(s), x^+ - x^-}
        ~ \leq ~ \| \Gradient f_p(s) \| \kern+1pt \| x^+ - x^- \|.
\]
By $\tau$-admissibility, $\|\Gradient f_p(s)\| \leq \tau f_p(s)/\|s - p\|$, and since $x^+, x^-, s \in B$, we have $\| x^+ - x^- \|/\|s - p\| \leq \diam(B)/\dist(p,B) \leq 1/(\tau \kappa)$. Thus,
\[
    f_p^+(B) - f_p^-(B)
        ~ \leq ~ \frac{\tau f_p(s)}{\| s - p \|} \kern+1pt \| x^+ - x^- \|
        ~ \leq ~ \frac{\tau f_p(s)}{\tau \kappa}
        ~ \leq ~ \frac{f_p^+(B)}{\kappa}.
\]
This implies that $f_p^+(B) \leq f_p^-(B) \kappa/ (\kappa - 1)$, establishing~(i).

To prove~(ii), consider any $x \in B$. By separation, $\dist(p, B) \geq \tau \kappa \cdot \diam(B)$. Combining this with $\tau$-admissibility and~(i), we have
\[
    \| \Gradient f_p(x) \|
          ~ \leq ~ \frac{\tau f_p(x)}{\|x - p\|}
          ~ \leq ~ \frac{\tau f_p^+(B)}{\dist(p, B)} 
          ~ \leq ~ \frac{\tau f_p^+(B)}{\tau \kappa \cdot \diam(B)}
          ~ = ~ \frac{f_p^+(B)}{\kappa \cdot \diam(B)}.
\]
This applies to any $x \in B$, thus establishing~(ii).

To prove~(iii), again consider any $x \in B$. By separation and admissibility, we have
\[
    \| \Hess f_p(x) \|
          ~ \leq ~ \frac{\tau^2 f_p(x)}{\|x - p\|^2}
          ~ \leq ~ \frac{\tau^2 f_p^+(B)}{\dist^2(p, B)} \\
          ~ \leq ~ \frac{f_p^+(B)}{(\kappa \cdot \diam(B))^2}.
\]
This applies to any $x \in B$, thus establishing~(iii).
\end{proof}

\subsection{Convexification and Ray Shooting} \label{s:convex-ray-shoot}

A set $\mathcal{F} = \{f_1, \ldots, f_m\}$ of $\tau$-admissible functions is called a \emph{$\tau$-admissible family of functions}. Let $\mathcal{F}_{\min}$ denote the associated lower-envelope function. In Lemma~\ref{l:ray-shoot-vertical} we showed that absolute $\eps$-AVR queries could be answered efficiently in a very restricted context. This will need to be generalized for the purposes of answering ANN queries, however. 

The main result of this section states that if the sites defining the distance functions are sufficiently well separated from a Euclidean ball, then (through convexification) $\eps$-AVR queries can be efficiently answered. The key idea is to map the ball and functions into the special structure required by Lemma~\ref{l:ray-shoot-vertical}, and to analyze how the mapping process affects the gradients and Hessians of the functions.

\begin{lemma} \label{l:ray-shoot-general}
\emph{(Convexification and Ray-Shooting)}
Consider a Euclidean ball $B \in \RE^d$ and a family of $\tau$-admissible distance functions $\mathcal{F} = \{f_1, \ldots, f_m\}$ over $B$ such that each associated site is $(2 \tau)$-separated from $B$. Given any $\eps > 0$, there exists a data structure that can answer relative $\eps$-AVR queries with respect to $\mathcal{F}_{\min}$ in time $O\big( \log \inv{\eps} \big)$ and storage $O\big( (\inv{\eps})^{d/2} \big)$.
\end{lemma}

\begin{proof} 
We will answer approximate vertical ray-shooting queries by a reduction to the data structure given in Lemma~\ref{l:ray-shoot-vertical}. In order to apply this lemma, we need to transform the problem into the canonical form prescribed by that lemma.

We may assume without loss of generality that $f_1$ is the function that defines the value of $f^-(B)$ among all the functions in $\mathcal{F}$. By Lemma~\ref{l:admiss-util}(i) (with $\kappa = 2$), $f_1^+(B) \leq 2 f_1^-(B)$. For all $i$, we may assume that $f_i^-(B) \leq 2 f_1^-(B)$ for otherwise this function is greater than $f_1$ throughout $B$, and hence it does not contribute to $\mathcal{F}_{\min}$. Under this assumption, it follows that $f_i^+(B) \leq 4 f_1^-(B)$.

In order to convert these functions into the desired form, define $h = 5 f_1^-(B)$, $r = \radius(B)$, and let $c \in \RE^d$ denote the center of $B$. Let $B_0$ be a unit ball centered at the origin, and for any $x \in B_0$, let $x' = r x + c$. Observe that $x \in B_0$ if and only if $x' \in B$. For each $i$, define the normalized distance function
\[
    g_i(x) 
        ~ = ~ \frac{f_i(x')}{h}.
\]
We assert that the normalized functions satisfy the following three properties. 
\begin{enumerate}
\item[(a)] $g_i^+(B_0) \leq 4/5$ and $g_i^-(B_0) \geq 1/5$

\item[(b)] $\|\Gradient g_i^+(B_0)\| \leq 1/2$

\item[(c)] $\|\Hess g_i^+(B_0)\| \leq 1/4$
\end{enumerate}
To establish (a), observe that for any $x \in B_0$,
\[
    g(x) 
        ~ \leq ~ \frac{f^+(B)}{h}
        ~ \leq ~ \frac{2 f^-(B)}{h}
        ~ \leq ~ \frac{4 f_1^-(B)}{h}
        ~  =   ~ \frac{4}{5}
    \quad\text{and}\quad
    g(x) 
        ~ \geq ~ \frac{f^-(B)}{h} 
        ~ \geq ~ \frac{f_1^-(B)}{h}
        ~  =   ~ \frac{1}{5}.
\]

In order to prove (b) and (c), observe that by the chain rule in differential calculus, $\Gradient g(x) = (r/h) \Gradient f(x')$ and $\Hess g(x) = (r^2/h) \Hess f(x')$. Since $B_0$ is a unit ball, $\diam(B_0) = 2$. Thus, by Lemma~\ref{l:admiss-util}(ii) (with $\kappa = 2$), we have 
\[
    \|\Gradient g(x)\|
        ~  =   ~ \frac{r}{h} \|\Gradient f(x')\|
        ~ \leq ~ \frac{r}{h} \frac{f^+(B)}{2(2r)}
        ~ \leq ~ \frac{1}{4},
\]
which establishes~(b). By Lemma~\ref{l:admiss-util}(iii),
\[
    \|\Hess g(x)\|
        ~  =   ~ \frac{r^2}{h} \|\Hess f(x')\|
        ~ \leq ~ \frac{r^2}{h} \frac{f^+(B)}{(2(2r))^2}
        ~ \leq ~ \frac{1}{16},
\]
which establishes~(c).

We now proceed to convexify these functions. To do this, define $\phi(x) = (1 - \|x\|^2)/8$. Observe that for any $x \in B_0$, $\phi(x) \in [0,1/8]$ and $\|\Gradient \phi(x)\| = \|x\|/4$ and $\Hess \phi(x)$ is the diagonal matrix $-(1/4)I$. 
Define
\[
    \conc{g}_i(x)
        ~ = ~ g_i(x) + \phi(x).
\]
It is easily verified that these functions satisfy the following properties.
\begin{enumerate}
\item[(a$'$)] $\conc{g}_i^{\kern+1pt +}(B_0) \leq 1$ and $\conc{g}_i^{\kern+1pt -}(B_0) \geq 1/5$

\item[(b$'$)] $\|\Gradient \conc{g}_i^{\kern+1pt +}(B_0)\| ~\leq~ \|\Gradient g_i^+(B_0)\| + \|\Gradient \phi^+(B_0)\| ~<~ 1$

\item[(c$'$)] $\|\Hess \conc{g}_i^{\kern+1pt +}(B_0)\| ~\leq~ \|\Hess g_i^+(B_0)\| - (1/4) ~\leq~ 0$
\end{enumerate}

By property (c$'$), these functions are concave over $B_0$. Given that $\conc{g}_i^{\kern+1pt -}(B_0) \geq 1/5$, in order to answer AVR queries to a relative error of $\eps$, it suffices to answer AVR queries to an absolute error of $\eps' = \eps/5$. Therefore, we can apply Lemma~\ref{l:ray-shoot-vertical} (using $\eps'$ in place of $\eps$) to obtain a data structure that answers relative $\eps$-AVR queries with respect to $\mathcal{F}_{\min}$ in time $O\big( \log \inv{\eps} \big)$ and storage $O\big( (\inv{\eps})^{d/2} \big)$, as desired.
\end{proof}

Armed with this tool, we are now in a position to describe the data structures for answering $\eps$-ANN queries for each of our applications, which we present in the subsequent sections.

\section{Answering ANN Queries for Scaling Distance Functions} \label{s:scaling-dist}

Recall that in a scaling distance we are given a convex body $K$ that contains the origin in its interior, and the distance from a query point $q$ to a site $p$ is defined to be zero if $p = q$ and otherwise it is the smallest $r$ such that $(q-p)/r \in K$.%
\footnote{This can be readily generalized to squared distances, that is, the smallest $r$ such that $(q-p)/\sqrt{r} \in K$. A relative error of $1+\eps$ in the squared distance, reduces to computing a $\sqrt{1+\eps}$ relative error in the original distance. Since $\sqrt{1+\eps} \approx (1+\eps/2)$ for small $\eps$, our approach can be applied but with a slightly smaller value of $\eps$. This generalizes to any constant power.} 
The body $K$ plays the role of a unit ball in a normed metric, but we do not require that the body be centrally symmetric. In this section we establish Theorem~\ref{thm:convexdf} by demonstrating a data structure for answering $\eps$-ANN queries given a set $S$ of $n$ sites, where each site $p_i$ is associated with a scaling distance whose unit ball is a fat, smooth convex body.

Before presenting the data structure, we present two preliminary results. The first, given in Section~\ref{s:avd-sep}, explains how to subdivide space into a number of regions, called \emph{cells}, that possess nice separation properties with respect to the sites. The second, given in Section~\ref{s:scaling-admiss}, presents key technical properties of scaling functions whose unit balls are fat and smooth.

\subsection{AVD and Separation Properties} \label{s:avd-sep}

In order to apply the convexification process, we will first subdivide space into regions, each of which satisfies certain separation properties with respect to the sites $S$. This subdivision results from a height-balanced variant of a quadtree, called a \emph{balanced box decomposition tree} (or BBD tree)~\cite{ArM00}. Each cell of this decomposition is either a quadtree box or the set-theoretic difference of two such boxes. Each leaf cell is associated with an auxiliary ANN data structure for the query points in the cell, and together the leaf cells subdivide all of $\RE^d$. 

The separation properties are essentially the same as those of the AVD data structure of \cite{AMM09a}. For any leaf cell $w$ of the decomposition, the sites can be partitioned into three subsets, any of which may be empty (see Figure~\ref{f:separation}(a)). First, a single site may lie within $w$. Second, a subset of sites, called the \emph{outer cluster}, is well-separated from the cell. Finally, there may be a dense cluster of points, called the \emph{inner cluster}, that lie within a ball $B_w$ that is well-separated from the cell. After locating the leaf cell containing the query point, the approximate nearest neighbor is computed independently for each of these subsets (by a method to be described later), and the overall closest is returned. The next lemma formalizes these separation properties. It follows easily from Lemma~{6.1} in~\cite{AMM09a}. Given a BBD-tree cell $w$ and a point $p \in \RE^d$, let $\dist(p,w)$ denote the minimum Euclidean distance from $p$ to any point in $w$.

\begin{figure*}[htbp]
  \centerline{\includegraphics[scale=0.40]{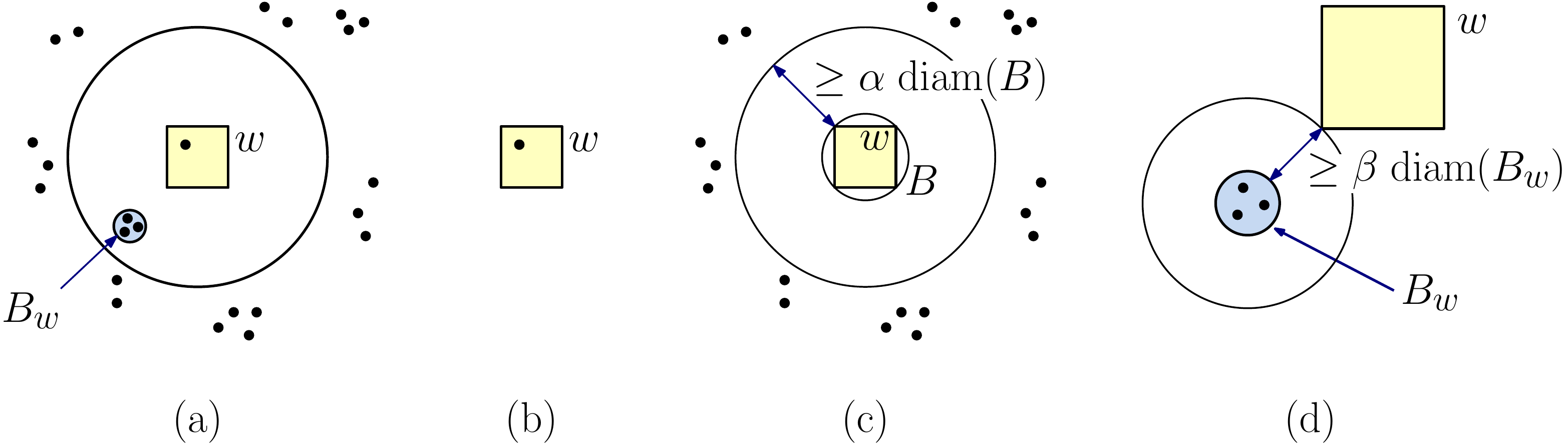}}
  \caption{Basic separation properties for Lemma~\ref{l:separation}.}
  \label{f:separation}
\end{figure*}

\begin{lemma}[Basic Separation Properties]
\label{l:separation}
Given a set $S$ of $n$ points in $\RE^d$ and real parameters $\alpha, \beta \ge 2$. It is possible to construct a BBD tree $T$ with $O(\alpha^d n \log \beta)$ nodes, whose leaf cells cover $\RE^d$ and for every leaf cell $w$, at least one of the following holds:
\begin{enumerate}
\item[$(i)$] Letting $p$ denote the closest site to the center of $w$, $p$ lies within $w$ and it is the only such site (see Figure~\ref{f:separation}(b)).

\item[$(ii)$] \emph{(outer cluster)} Letting $B$ denote the smallest Euclidean ball enclosing $w$, $\dist(p, B) \ge \alpha \cdot \diam(B)$ (see Figure~\ref{f:separation}(c)).

\item[$(iii)$] \emph{(inner cluster)} There exists a ball $B_w$ associated with $w$ such that $\dist(B_w, w) \ge \beta \cdot \diam(B_w)$ and $p \in B_w$ (see Figure~\ref{f:separation}(d)).
\end{enumerate}

Furthermore, it is possible to compute the tree $T$ in total time $O(\alpha^d n \log n \kern+1pt \log \beta)$, and the leaf cell containing a query point can be located in time $O(\log (\alpha n) + \log \log \beta)$.
\end{lemma}

\subsection{Admissibility for Scaling Distances} \label{s:scaling-admiss}

In this section we explore how properties of the unit ball affect the effectiveness of convexification. Recall from Section~\ref{s:convexify} that convexification relies on the admissibility of the distance function, and we show here that this will be guaranteed if unit balls are fat, well centered, and smooth.

Given a convex body $K$ and a parameter $0 < \gamma \leq 1$, we say that $K$ is \emph{centrally $\gamma$-fat} if there exist Euclidean balls $B$ and $B'$ centered at the origin, such that $B \subseteq K \subseteq B'$, and $\radius(B) / \radius(B') \ge \gamma$. Given a parameter $0 < \sigma \leq 1$, we say that $K$ is \emph{$\sigma$-smooth} if for every point $x$ on the boundary of $K$, there exists a closed Euclidean ball of diameter $\sigma \cdot \diam(K)$ that lies within $K$ and has $x$ on its boundary. We say that a scaling distance function is a \emph{$(\gamma,\sigma)$-distance} if its associated unit ball $B$ is both centrally $\gamma$-fat and $\sigma$-smooth.

In order to employ convexification for scaling distances, it will be useful to show that smoothness and fatness imply that the associated distance functions are admissible. This is encapsulated in the following lemma. It follows from a straightforward but rather technical exercise in multivariate differential calculus. We include a complete proof in Section~\ref{s:deferred}.

\begin{restatable}{lemma}{scalingadmiss}
\label{l:scaling-admiss}
Given positive reals $\gamma$ and $\sigma$, let $f_p$ be a $(\gamma,\sigma)$-distance over $\RE^d$ scaled about some point $p \in \RE^d$. There exists $\tau$ (a function of $\gamma$ and $\sigma$) such that $f_p$ is $\tau$-admissible.
\end{restatable}

Our results on $\eps$-ANN queries for scaling distances will be proved for any set of sites whose associated distance functions (which may be individual to each site) are all $(\gamma,\sigma)$-distances for fixed $\gamma$ and $\sigma$. Our results on the Minkowski and Mahalanobis distances thus arise as direct consequences of the following easy observations.

\begin{lemma} ~
\label{l:gamma-sigma}
\begin{itemize}
\item[$(i)$] For any positive real $k > 1$, the Minkowski distance $\ell_k$ in $\RE^d$ is a $(\gamma,\sigma)$-distance, where $\gamma$ and $\sigma$ are functions of $k$ and $d$. 

This applies to multiplicatively weighted Minkowski distances as well.
     
\item[$(ii)$] The Mahalanobis distance defined by a matrix $M_p$ in $\RE^d$ is a $(\gamma,\sigma)$-distance, where $\gamma$ and $\sigma$ are functions of $M_p$'s minimum and maximum eigenvalues.
\end{itemize} 
\end{lemma}

\subsection{ANN Data Structure for Scaling Functions} \label{s:scaling-data-struct}

Let us return to the discussion of how to answer $\eps$-ANN queries for a family of $(\gamma,\sigma)$-distance functions. By Lemma~\ref{l:scaling-admiss}, such functions are $\tau$-admissible, where $\tau$ depends only on $\gamma$ and $\sigma$. 

We begin by building an $(\alpha,\beta)$-AVD over $\RE^d$ by invoking Lemma~\ref{l:separation} for $\alpha = 2 \tau$ and $\beta = 10 \tau/\eps$. (These choices will be justified below.) For each leaf cell $w$, the nearest neighbor of any query point $q \in w$ can arise from one of the three cases in the lemma. Case~(i) is trivial since there is just one point. 

Case~(ii) (\emph{outer cluster}) can be solved easily by reduction to Lemma~\ref{l:ray-shoot-general}. Recall that we have a BBD-tree leaf cell $w$, and the objective is to compute an $\eps$-ANN from among the points of the outer cluster, that is, a set whose sites are at Euclidean distance at least $\alpha \cdot \diam(w)$ from $w$. Let $B$ denote the smallest Euclidean ball enclosing $w$, and let $\mathcal{F}$ be the family of distance functions associated with the sites of the outer cluster. Since $\alpha = 2 \tau$, $B$ is $(2 \tau)$-separated from the points of the outer cluster. By Lemma~\ref{l:ray-shoot-general}, we can answer $\eps$-AVR queries with respect to $\mathcal{F}_{\min}$, and this is equivalent to answering $\eps$-ANN queries with respect to the outer cluster. The query time is $O(\log \inv{\eps})$ and the storage is $O((\inv{\eps})^{d/2})$.

All that remains is case~(iii) (\emph{inner cluster}). Recall that these sites lie within a ball $B_w$ such that $\dist(B_w, w) \ge \beta \cdot \diam(B_w)$. In approximate Euclidean nearest-neighbor searching, a separation as large as $\beta$ would allow us to replace all the points of $B_w$ with a single representative site, but this is not applicable when different sites are associated with different scaling distance functions. We will show instead that queries can be answered by partitioning the query space into a small number of regions such that Lemma~\ref{l:ray-shoot-general} can be applied to each region. Let $\{p_1, \ldots, p_m\}$ denote the sites lying within $B_w$, and let $\mathcal{F} = \{f_1, \ldots, f_m\}$ denote the associated family of $(\gamma,\sigma)$-distance functions. 

Let $p'$ be the center of $B_w$, and for $1 \leq i \leq m$, define the \emph{perturbed distance function} $f'_i(x) = f_i(x + p_i - p')$ to be the function that results by moving $p_i$ to $p'$ without altering the unit metric ball. Let $\mathcal{F'}$ denote the associated family of distance functions. Our next lemma shows that this perturbation does not significantly alter the relative function values.

\begin{lemma} \label{l:perturbation}
Let $p \in \RE^d$ be the site of a $\tau$-admissible distance function $f$. Let $B$ be a ball containing $p$ and let $x$ be a point that is $\beta$-separated from $B$ for $\beta \geq 2\tau$. Letting $p'$ denote $B$'s center, define $f'(x) = f(x + p - p')$. Then 
\[
    \frac{|f'(x) - f(x)|}{f(x)}
        ~ \leq ~ \frac{2\tau}{\beta}.
\]

\end{lemma}

\begin{proof}
Define $B_x$ to be the translate of $B$ whose center coincides with $x$. Since $p$ and $p'$ both lie within $B$, $x$ and $x + p - p'$ both lie within $B_x$. Let $\kappa = \beta/\tau$. Since $x$ and $B$ are $\beta$-separated, $p'$ and $B_x$ are also $\beta$-separated. Equivalently, they are $(\tau\kappa)$-separated. Because $\kappa \geq 2$, $\kappa/(\kappa-1) \leq (1 + 2/\kappa)$. Because $f'$ has the same unit metric ball as $f$, it is also $\tau$-admissible, and so by Lemma~\ref{l:admiss-util} 
\[
    {f'}^+(B_x) 
        ~ \leq ~ \frac{\kappa}{\kappa - 1} {f'}^-(B_x)  
        ~ \leq ~ \left(1 + \frac{2}{\kappa}\right) {f'}^-(B_x)
        ~   =  ~ \left(1 + \frac{2\tau}{\beta}\right) {f'}^-(B_x).
\]
Letting $x' = x - (p - p')$, we have $f(x) = f'(x')$. Clearly $x' \in B_x$. Let us assume that $f'(x) \geq f(x)$. (The other case is similar.) We have
\begin{align*}
    f'(x) - f(x)
        & ~   =  ~ f'(x) - f'(x')
          ~ \leq ~ {f'}^+(B_x) - {f'}^-(B_x) \\
        & ~ \leq ~ \frac{2\tau}{\beta} {f'}^-(B_x)
          ~ \leq ~ \frac{2\tau}{\beta} f'(x') 
          ~   =  ~ \frac{2\tau}{\beta} f(x),
\end{align*}
which implies the desired inequality.
\end{proof}

Since every point $x \in w$ is $\beta$-separated from $B_w$, by applying this perturbation to every function in $\mathcal{F}$, we alter relative errors by at most $2 \tau/\beta$. By selecting $\beta$ so that $(1 + 2\tau/\beta)^2 \leq 1 + \eps/2$, we assert that the total error is at most $\eps/2$. To see this, consider any query point $x$, and let $f_i$ be the function that achieves the minimum value for $\mathcal{F}_{\min}(x)$, and let $f'_j$ be the perturbed function that achieves the minimum value for $\mathcal{F'}_{\min}(x)$. Then
\begin{align*}
    f_j(x)
        & ~ \leq ~ \left(1 + \frac{2\tau}{\beta}\right) f'_j(x)
          ~ \leq ~ \left(1 + \frac{2\tau}{\beta}\right) f'_i(x) \\
        & ~ \leq ~ \left(1 + \frac{2\tau}{\beta}\right)^2 f_i(x)
          ~ \leq ~ \left(1 + \frac{\eps}{2}\right) f_i(x).
\end{align*}
It is easy to verify that for all sufficiently small $\eps$, our choice of $\beta = 10 \tau/\eps$ satisfies this condition (and it is also at least $2 \tau$ as required by the lemma).

We can now explain how to answer $\eps$-ANN queries for the inner cluster. Consider the sites of the inner cluster, which all lie within $B_w$ (see Figure~\ref{f:inner}(a)). We apply Lemma~\ref{l:perturbation} to produce the perturbed family $\mathcal{F'}$ of $\tau$-admissible functions (see Figure~\ref{f:inner}(b)).

\begin{figure*}[htbp]
  \centerline{\includegraphics[scale=0.40]{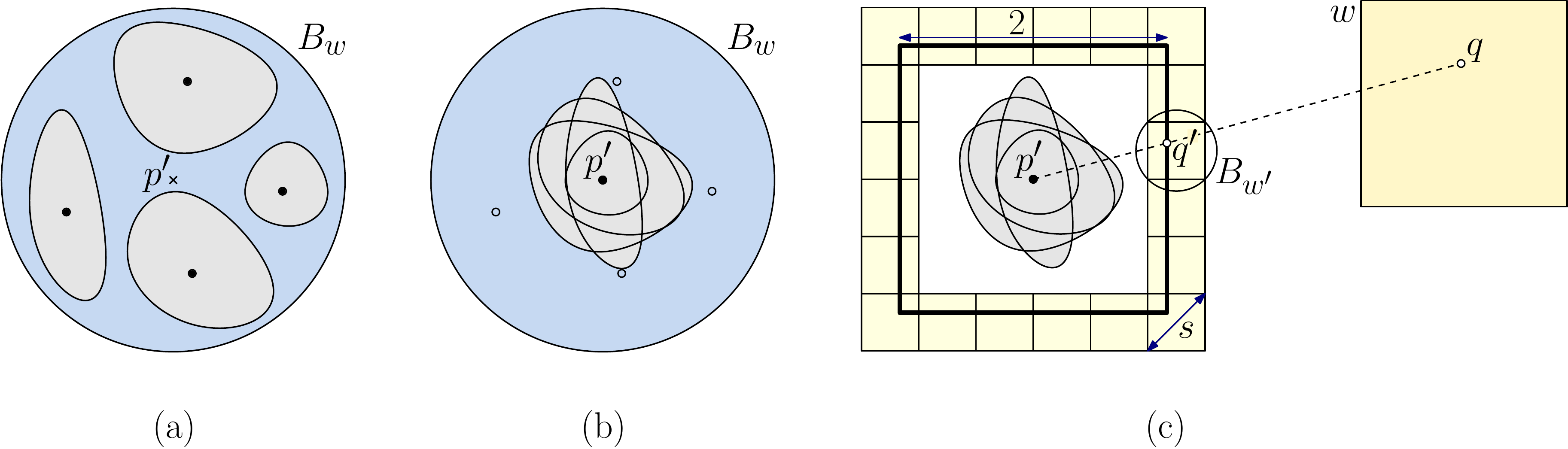}}
  \caption{(a) Inner-cluster sites with their respective distance functions, (b) their perturbation to a common site $p'$, and (c) the reduction to Lemma~\ref{l:ray-shoot-general}.}
  \label{f:inner}
\end{figure*}

Since these are all scaling distance functions, the nearest neighbor of any query point $q \in \RE^d$ (irrespective of whether it lies within $w$) is the same for every point on the ray from $p'$ through $q$. Therefore, it suffices to evaluate the answer to the query for any single query point $q'$ on this ray. In particular, let us fix a hypercube of side length $2$ centered at $p'$ (see Figure~\ref{f:inner}(c)). We will show how to answer $(\eps/3)$-AVR queries for points on the boundary of this hypercube with respect to $\mathcal{F'}$. A general query will then be answered by computing the point where the ray from $p'$ to the query point intersects the hypercube's boundary and returning the result of this query. The total error with respect to the original functions will be at most $(1 + \eps/2)(1 + \eps/3)$, and for all sufficiently small $\eps$, this is at most $1 + \eps$, as desired.

All that remains is to show how to answer $(\eps/3)$-AVR queries for points on the boundary of the hypercube. Let $s = 1/(2\tau + 1)$, and let $W$ be a set of hypercubes of diameter $s$ that cover the boundary of the hypercube of side length $2$ centered at $p'$ (see Figure~\ref{f:inner}(c)). The number of such boxes is $O(\tau^{d-1})$. For each $w' \in W$, let $B_{w'}$ be the smallest ball enclosing $w'$. Each point on the hypercube is at distance at least $1$ from $p'$. For each $w' \in W$, we have $\dist(p', B_{w'}) \geq 1 - s = 2 \tau \cdot \diam(B_{w'})$, implying that $p'$ and $B_{w'}$ are $(2 \tau)$-separated. Therefore, by Lemma~\ref{l:ray-shoot-general} there is a data structure that can answer $(\eps/3)$-AVR queries with respect to the perturbed distance functions $\mathcal{F'}_{\min}$ in time $O(\log \inv{\eps})$ with storage $O((\inv{\eps})^{d/2})$.

In summary, a query is answered by computing the ray from $p'$ through $q$, and determining the unique point $q'$ on the boundary of the hypercube that is hit by this ray. We then determine the hypercube $w'$ containing $q'$ in constant time and invoke the associated data structure for answering $(\eps/3)$-AVR queries with respect to $\mathcal{F'}$. The total storage needed for all these structures is $O(\tau^{d-1}/\eps^{d/2})$. For any query point, we can determine which of these data structures to access in $O(1)$ time. Relative to the case of the outer cluster, we suffer only an additional factor of $O(\tau^{d-1})$ to store these data structures. 

Under our assumption that $\gamma$ and $\sigma$ are constants, it follows that both $\tau$ and $\alpha$ are constants and $\beta$ is $O(1/\eps)$. By Lemma~\ref{l:separation}, the total number of leaf nodes in the $(\alpha,\beta)$-AVD is $O(n \log \inv{\eps})$. Combining this with the $O(1/\eps^{d/2})$ space for the data structure to answer queries with respect to the outer cluster and $O(\tau^{d-1}/\eps^{d/2})$ overall space for the inner cluster, we obtain a total space of $O((n \log\inv\eps)/\eps^{d/2})$. The query time is simply the combination of the $O(\log (\alpha n) + \log\log \beta) = O(\log n + \log\log \inv{\eps})$ time to locate the leaf cell (by Lemma~\ref{l:separation}), and the $O(\log \inv{\eps})$ time to answer $O(\eps)$-AVR queries. The total query time is therefore $O(\log \frac{n}{\eps})$, as desired. This establishes Theorem~\ref{thm:convexdf}.

\section{Answering ANN Queries for Bregman Divergences} \label{s:bregman-ann}

In this section we demonstrate how to answer $\eps$-ANN queries for a set of $n$ sites over a Bregman divergence. We assume that the Bregman divergence is defined by a strictly convex, twice-differentiable function $F$ over an open convex domain $\XX \subseteq \RE^d$. As mentioned in the introduction, given a site $p$, we interpret the divergence $D_F(x,p)$ as a distance function of $x$ about $p$, that is, analogous to $f_p(x)$ for scaling distances. Thus, gradients and Hessians are defined with respect to the variable $x$. Our results will be based on the assumption that the divergence is $\tau$-admissible for a constant $\tau$. This will be defined formally in the following section.

\subsection{Measures of Bregman Complexity} \label{s:bregman-complexity}

In Section~\ref{s:intro} we introduced the concepts of similarity and asymmetry for Bregman divergences. We can extend the notion of admissibility to Bregman divergences by defining a Bregman divergence $D_F$ to be \emph{$\tau$-admissible} if the associated distance function $f_p(\cdot) = D_F(\cdot,p)$ is $\tau$-admissible.

It is natural to ask how the various criteria of Bregman complexity (asymmetry, similarity, and admissibility) relate to each other. For the sake of relating admissibility with asymmetry, it will be helpful to introduce a directionally-sensitive variant of admissibility. Given $f_p$ and $\tau$ as above, we say that $f_p$ is \emph{directionally $\tau$-admissible} if for all $x \in \XX$, $\ang{\Gradient f_p(x), x - p} \leq \tau f_p(x)$. (Note that only the gradient condition is used in this definition.) The following lemma provides some comparisons. The proof is rather technical and has been deferred to Section~\ref{s:deferred}.

\begin{restatable}{lemma}{bregmancomparisons}
\label{l:bregman-comparisons}
Given an open convex domain $\XX \subseteq \RE^d$:
\begin{enumerate}
\item[$(i)$] Any $\mu$-similar Bregman divergence over $\XX$ is $2 \mu$-admissible.

\item[$(ii)$] Any $\mu$-admissible Bregman divergence over $\XX$ is directionally $\mu$-admissible.

\item[$(iii)$] A Bregman divergence over $\XX$ is $\mu$-asymmetric if and only if it is directionally $(1+\mu)$-admissible.
\end{enumerate}
\end{restatable}

Note that claim~(i) is strict since the Bregman divergence $D_F$ defined by $F(x) = x^4$ over $\XX = \RE$ is not $\mu$-similar for any $\mu$, but it is 4-admissible. We do not know whether claim~(ii) is strict, but we conjecture that it is.

\subsection{ANN Data Structure for Bregman Divergences} \label{s:bregman-data-struct}

Let us return to the discussion of how to answer $\eps$-ANN queries for a $\tau$-admissible Bregman divergence over a domain $\XX$. Because any distance function that is $\tau$-admissible is $\tau'$-admissible for any $\tau' \geq \tau$, we may assume that $\tau \geq 1$.%
\footnote{Indeed, it can be shown that any distance function that is convex, as Bregman divergences are, cannot be $\tau$-admissible for $\tau < 1$.}
We begin by building an $(\alpha,\beta)$-AVD over $\RE^d$ by invoking Lemma~\ref{l:separation} for $\alpha = 2 \tau$ and $\beta = 4 \tau^2/\eps$. (These choices will be justified below.) For each leaf cell $w$, the nearest neighbor of any query point $q \in w$ can arise from one of the three cases in the lemma. Cases~(i) and~(ii) are handled in exactly the same manner as in Section~\ref{s:scaling-data-struct}. (Case~(i) is trivial, and case~(ii) applies for any $\tau$-admissible family of functions.)

It remains to handle case~(iii), the \emph{inner cluster}. Recall that these sites lie within a ball $B_w$ such that $\dist(B_w, w) \ge \beta \cdot \diam(B_w)$. We show that as a result of choosing $\beta$ sufficiently large, for any query point in $w$ the distance from all the sites within $B_w$ are sufficiently close that we may select any of these sites as the approximate nearest neighbor. This is a direct consequence of the following lemma. The proof has been deferred to Section~\ref{s:deferred}.

\begin{restatable}{lemma}{bregmaninner}
\label{l:bregman-inner}
Let $D$ be a $\tau$-admissible Bregman divergence and let $0 < \eps \leq 1$. Consider any leaf cell $w$ of the $(\alpha,\beta)$-AVD, where $\beta \geq 4 \tau^2/\eps$. Then, for any $q \in w$ and points $p, p' \in B_w$
\[
    \frac{|D(q, p) - D(q, p')|}{D(q,p)}
        ~ \leq ~\eps.
\]
\end{restatable}

Under our assumption that $\tau$ is a constant, $\alpha$ is a constant and $\beta$ is $O(1/\eps)$. The analysis is similar to the case for scaling distances. By Lemma~\ref{l:separation}, the total number of leaf nodes in the $(\alpha,\beta)$-AVD is $O(n \log \inv{\eps})$. We require only one representative for cases~(i) and~(iii), and as in Section~\ref{s:scaling-data-struct}, we need space $O(1/\eps^{d/2})$ to handle case~(ii). The query time is simply the combination of the $O(\log (\alpha n) + \log\log \beta) = O(\log n + \log\log \inv{\eps})$ time to locate the leaf cell (by Lemma~\ref{l:separation}), and the $O(\log \inv{\eps})$ time to answer $O(\eps)$-AVR queries for case~(ii). The total query time is therefore $O(\log \frac{n}{\eps})$, as desired. This establishes Theorem~\ref{thm:bregman}.

\section{Deferred Technical Details} \label{s:deferred}

In this section we present a number of technical results and proofs, which have been deferred from the main presentation.

\subsection{On Vertical Ray Shooting}

In this section we present a proof of Lemma~\ref{l:ray-shoot-vertical} from Section~\ref{s:ray-shoot}, which shows how to answer approximate vertical ray-shooting queries for the lower envelope of concave functions in a very restricted context.

{\rayshootvertical*}

We will follow the strategy presented in \cite{AFM18a} for answering $\eps$-ANN queries. It combines (1) a data structure for answering approximate central ray-shooting queries, in which the rays originate from a common point and (2) an approximation-preserving reduction from vertical to central ray-shooting queries~\cite{AFM17a}. 

Let $K$ denote a closed convex body that is represented as the intersection of a finite set of halfspaces. We assume that $K$ is centrally $\gamma$-fat for some constant $\gamma$ (recall the definition from Section~\ref{s:scaling-admiss}). An \emph{$\eps$-approximate central ray-shooting query} ($\eps$-ACR query) is given a query ray that emanates from the origin and returns the index of one of $K$'s bounding hyperplanes $h$ whose intersection with the ray is within distance $\eps \cdot \diam(K)$ of the true contact point with $K$'s boundary. We will make use of the following result, which is paraphrased from \cite{AFM17a}.

\begin{description}
\item[Approximate Central Ray-Shooting:] Given a convex polytope $K$ in $\RE^d$ that is centrally $\gamma$-fat for some constant $\gamma$ and an approximation parameter $0 < \eps \le 1$, there is a data structure that can answer $\eps$-ACR queries in time $O(\log \inv{\eps})$ and storage $O(1/\eps^{(d-1)/2})$.
\end{description}

As in Section~{4} of \cite{AFM17a}, we can employ a projective transformation that converts vertical ray shooting into central ray shooting. While the specific transformation presented there was tailored to work for a set of hyperplanes that are tangent to a paraboloid, a closer inspection reveals that the reduction can be generalized (with a change in the constant factors) provided that the following quantities are all bounded above by a constant: (1) the diameter of the domain of interest, (2) the difference between the maximum and minimum function values throughout this domain, and (3) the absolute values of the slopes of the hyperplanes (or equivalently, the norms of the gradients of the functions defined by these hyperplanes). This projective transformation produces a convex body in $\RE^{d+1}$ that is centrally $\gamma$-fat for some constant $\gamma$, and it preserves relative errors up to a constant factor. 

Therefore, by applying this projective transformation, we can reduce the problem of answering $\eps$-AVR queries in dimension $d$ for the lower envelope of a set of linear functions to the aforementioned ACR data structure in dimension $d+1$. The only remaining issue is that the functions of $\mathcal{F}$ are concave, not necessarily linear. Thus, the output of the reduction is a convex body bounded by curved patches, not a polytope. We address this by applying Dudley's Theorem~\cite{Dud74} to produce a polytope that approximates this convex body to an absolute Hausdorff error of $\eps/2$. (In particular, Dudley's construction samples $O(1/\eps^{d/2})$ points on the boundary of the convex body, and forms the approximation by intersecting the supporting hyperplanes at each of these points.) We then apply the ACR data structure to this approximating polytope, but with the allowed error parameter set to $\eps/2$. The combination of the two errors, results in a total allowed error of $\eps$. 

In order to obtain a witness, each sample point from Dudley's construction is associated with the function(s) that are incident to that point. We make the general position assumption that no more than $d+1$ functions can coincide at any point on the lower envelope of $\mathcal{F}$, and hence each sample point is associated with a constant number of witnesses. The witness produced by the ACR data structure will be one of the bounding hyperplanes. We check each of the functions associated with the sample point that generated this hyperplane, and return the index of the function having the smallest function value.

\subsection{On Admissibility and Scaling Functions}

Next, we present a proof of Lemma~\ref{l:scaling-admiss} from Section~\ref{s:scaling-admiss}, which relates the admissibility of a scaling distance function to the fatness and smoothness of the associated metric ball.

{\scalingadmiss*}

\begin{proof}
For any point $x \in \RE^d$, we will show that (i) $\|\Gradient f_p(x)\| \cdot \|x-p\| \le f_p(x) /\gamma$ and (ii) $\|\Hess f_p(x)\| \cdot \|x-p\|^2 \le 2 f_p(x) / (\sigma \gamma^3)$ . It will follow that $f_p$ is $\tau$-admissible for $\tau = \sqrt{2 / (\sigma \gamma^3)}$.

Let $K$ denote the unit metric ball associated with $f_p$ and let $K'$ denote the scaled copy of $K$ that just touches the point $x$. Let $r$ be the unit vector in the direction $p x$ (we refer to this as the \emph{radial direction}), and let $n$ be the outward unit normal vector to the boundary of $K'$ at $x$. (Throughout the proof, unit length vectors are defined in the Euclidean sense.) As $K'$ is centrally $\gamma$-fat, it is easy to see that the cosine of the angle between $r$ and $n$, that is, $\ang{r, n}$, is at least $\gamma$. As the boundary of $K'$ is the level surface of $f_p$, it follows that $\Gradient f_p(x)$ is directed along $n$. To compute the norm of the gradient, note that
\[
    \ang{\Gradient f_p(x), r} 
        ~ = ~ \lim_{\delta \rightarrow 0} \frac{f_p(x + \delta r) - f_p(x)}{\delta}.
\]
As $f_p$ is a scaling distance function, it follows that 
\[
    f_p(x + \delta r) - f_p(x) 
        ~ = ~ \frac{\delta}{\|x - p\|} f_p(x).
\]
Thus,
\[
    \ang{\Gradient f_p(x), r} 
        ~ = ~ \frac{f_p(x)}{\|x - p\|}.
\]
Recalling that $\ang{r, n} \geq \gamma$, we obtain
\[
    \|\Gradient f_p(x)\| 
        ~ \leq ~ \frac{f_p(x)}{\gamma \|x - p\|}.
\]
Thus $\|\Gradient f_p(x)\| \cdot \|x-p\| \le f_p(x) /\gamma$, as desired.

We next bound the norm of the Hessian $\Hess f_p(x)$. As the Hessian matrix is positive semidefinite, recall that it has a full set of independent eigenvectors that are mutually orthogonal, and its norm equals its largest eigenvalue. Because $f_p$ is a scaling distance function, it changes linearly along the radial direction. Therefore, one of the eigenvectors of $\Hess f_p(x)$ is in direction $r$, and the associated eigenvalue is 0 (see Figure~\ref{f:hessian}). It follows that the remaining eigenvectors all lie in a subspace that is orthogonal to $r$. In particular, the eigenvector associated with its largest eigenvalue must lie in this subspace. Let $u$ denote such an eigenvector of unit length, and let $\lambda$ denote the associated eigenvalue. 

\begin{figure}[htbp]
  \centerline{\includegraphics[scale=0.60]{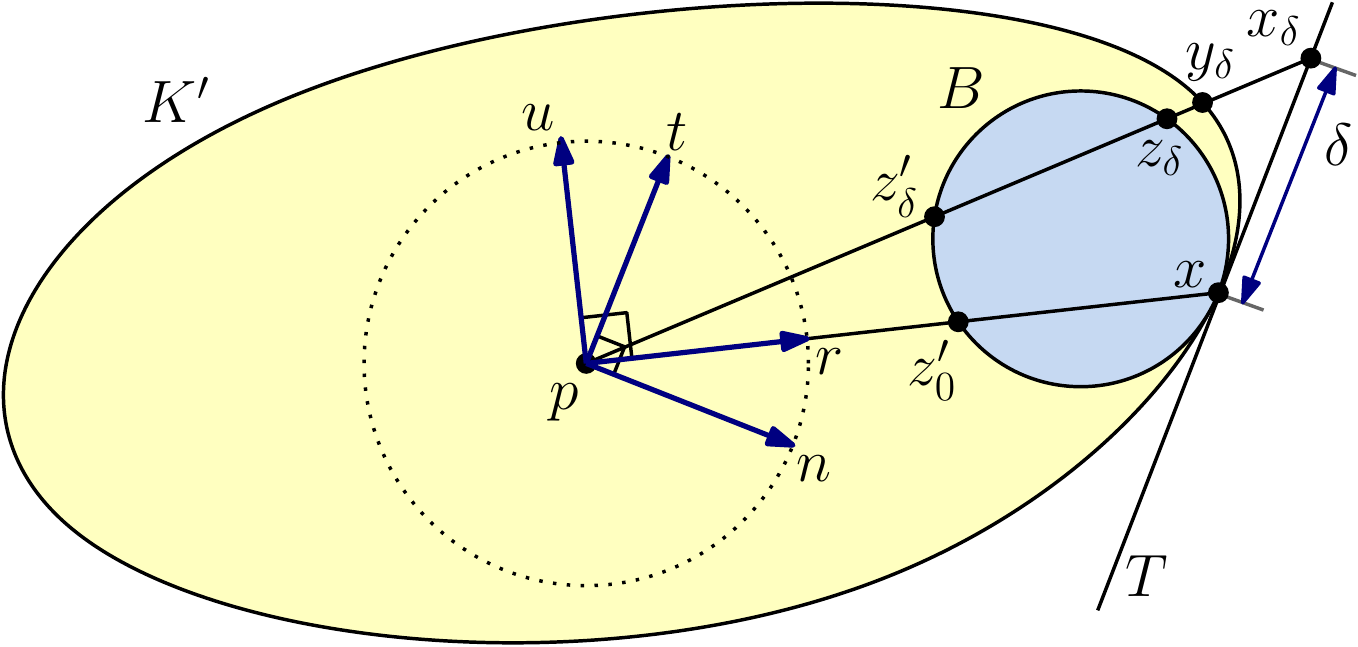}}
  \caption{Proof of Lemma~\ref{l:scaling-admiss}.}
  \label{f:hessian}
\end{figure}

Note that $\lambda$ is the second directional derivative of $f_p$ in the direction $u$. In order to bound $\lambda$, we find it convenient to first bound the second directional derivative of $f_p$ in a slightly different direction. Let $T$ denote the hyperplane tangent to $K'$ at point $x$. We project $u$ onto $T$ and let $t$ denote the resulting vector scaled to have unit length. We will compute the second directional derivative of $f_p$ in the direction $t$. Let $\lambda_t$ denote this quantity. In order to relate $\lambda_t$ with $\lambda$, we write $t$ as $\ang{t, r} r + \ang{t, u} u$. Since $r$ and $u$ are mutually orthogonal eigenvectors of $\Hess f_p(x)$, by elementary linear algebra, it follows that $\lambda_t = \ang{t, r}^2 \lambda_r + \ang{t, u}^2 \lambda_u$, where $\lambda_r$ and $\lambda_u$ are the eigenvalues associated with $r$ and $u$, respectively. Since $\lambda_r = 0$, $\lambda_u = \lambda$, and $\ang{t, u} = \ang{r, n} \geq \gamma$, we have $\lambda_t \geq \gamma^2 \lambda$, or equivalently, $\lambda \le \lambda_t / \gamma^2$. In the remainder of the proof, we will bound $\lambda_t$, which will yield the desired bound on $\lambda$.

Let $x_{\delta} = x + \delta t$ and $\psi(\delta) = f_p(x_{\delta})$. Clearly $\lambda_t = \psi''(0)$. Using the Taylor series and the fact that $\psi'(0) = \ang{\Gradient f_p(x), t} = 0$, it is easy to see that
\[
    \psi''(0) 
        ~ = ~ 2 \cdot \lim_{\delta \rightarrow 0} \frac{\psi(\delta) - \psi(0)}{\delta^2}.
\]
Letting $y_{\delta}$ denote the intersection point of the segment $\overline{px_{\delta}}$ with the boundary of $K'$, and observing that both $x$ and $y_{\delta}$ lie on $\partial K'$ (implying that $f_p(x) = f_p(y_{\delta})$), we have 
\[
    \psi(\delta) 
        ~ = ~ f_p(x_{\delta}) 
        ~ = ~ \frac{\|x_{\delta}- p\|}{\|y_{\delta} - p\|} f_p(x),
\]
and thus
\[
    \psi(\delta) - \psi(0)
        ~ = ~ \frac{\|x_{\delta} - p\| - \|y_{\delta} - p\|}{\|y_{\delta} - p\|} f_p(x) 
      ~ = ~ \frac{\|x_{\delta} - y_{\delta}\|}{\|y_{\delta} - p\|} f_p(x).
\]
It follows that
\[
    \psi''(0) 
        ~ = ~ 2 \cdot \lim_{\delta \rightarrow 0} \frac{1}{\delta^2} 
                \frac{\|x_{\delta} - y_{\delta}\|}{\|y_{\delta} - p\|} f_p(x)
        ~ = ~ \frac{2 f_p(x)}{\|x - p\|} \cdot \lim_{\delta \rightarrow 0} 
        \frac{\|x_{\delta} - y_{\delta}\|}{\delta^2}.
\]

We next compute this limit. Let $B \subset K'$ denote the maximal ball tangent to $K'$ at $x$ and let $R$ denote its radius. As $K'$ is $\sigma$-smooth, we have that
\[
    R
        ~ \geq ~ \frac{\sigma}{2} \cdot \diam(K') 
        ~ \geq ~ \frac{\sigma}{2} \cdot \|x - p\|.
\]
Consider the line passing through $p$ and $x_{\delta}$. For sufficiently small $\delta$, it is clear that this line must intersect the boundary of the ball $B$ at two points. Let $z_{\delta}$ denote the intersection point closer to $x_{\delta}$ and $z'_{\delta}$ denote the other intersection point. Clearly, $\|x_{\delta} - y_{\delta}\| \le \|x_{\delta} - z_{\delta}\|$ and, by the power of the point theorem, we have 
\[
    \delta^2 
        ~ = ~ \|x_{\delta} - x\|^2 
        ~ = ~ \|x_{\delta} - z_{\delta}\| \cdot \|x_{\delta} - z'_{\delta}\|.
\]
It follows that
\[
    \frac{\|x_{\delta} - y_{\delta}\|}{\delta^2}
        ~ \le ~ \frac{\|x_{\delta} - z_{\delta}\|}{\delta^2}
        ~ = ~ \frac{1}{\|x_{\delta} - z'_{\delta}\|}.
\]
Thus
\[
    \lim_{\delta \rightarrow 0} \frac{\|x_{\delta} - y_{\delta}\|}{\delta^2}
      ~ \leq ~ \lim_{\delta \rightarrow 0} \frac{1}{\|x_{\delta} - z'_{\delta}\|}
      ~ = ~ \frac{1}{\|x - z'_0\|},
\]
where $z'_0$ denotes the point of intersection of the line passing through $p$ and $x$ with the boundary of $B$. Since the cosine of the angle between this line and the diameter of ball $B$ at $x$ equals $\ang{r, n}$, which is at least $\gamma$, we have $\|x - z'_0\| \ge 2 \gamma R$. It follows that
\[
    \lim_{\delta \rightarrow 0} \frac{\|x_{\delta} - y_{\delta}\|}{\delta^2}
        ~ \leq ~ \frac{1}{2 \gamma R}
        ~ \leq ~ \frac{1}{\sigma \gamma \|x - p\|}.
\]

Substituting this bound into the expression found above for $\lambda_t$, we obtain
\[
    \lambda_t 
        ~ = ~ \psi''(0)
        ~ \le ~ \frac{2 f_p(x)}{\sigma \gamma \|x - p\|^2}.
\]
Recalling that $\lambda \le \lambda_t / \gamma^2$, we have
\[
    \lambda 
        ~ \leq ~ \frac{2 f_p(x)}{\sigma \gamma^3 \|x - p\|^2},
\]
which implies that $\|\Hess f_p(x)\| \cdot \|x-p\|^2 \le 2 f_p(x) / (\sigma \gamma^3)$. This completes the proof.
\end{proof}

\subsection{On Bregman Divergences}

The following lemma provides some properties of Bregman divergences, which will be used later. Throughout, we assume that a Bregman divergence is defined by a strictly convex, twice-differentiable function $F$ over an open convex domain $\XX \subseteq \RE^d$. Given a site $p$, we interpret the divergence $D_F(x,p)$ as a distance function of $x$ about $p$, and so gradients and Hessians are defined with respect to the variable $x$. The following lemma provides a few useful observations regarding the Bregman divergence. We omit the proof since these all follow directly from the definition of Bregman divergence. Observation~(i) is related to the \emph{symmetrized Bregman divergence} \cite{AMV13}. Observation~(ii), known as the \emph{three-point property} \cite{BNN10}, generalizes the law of cosines when the Bregman divergence is the Euclidean squared distance.

\begin{lemma}
\label{l:bregman-util}
Given any Bregman divergence $D_F$ defined over an open convex domain $\XX$, and points $q, p, p' \in \XX$:
\begin{enumerate}
\item[$(i)$] $D_F(q, p) + D_F(p, q) = \ang{\Gradient F(q) - \Gradient F(p), q - p}$

\item[$(ii)$] $D_F(q, p') + D_F(p', p) = D_F(q, p)  + \ang{q - p', \Gradient F(p) - \Gradient F(p')}$

\item[$(iii)$] $\Gradient D_F(q, p) = \Gradient F(q) - \Gradient F(p)$

\item[$(iv)$] $\Hess D_F(q, p) = \Hess F(q)$.
\end{enumerate}
In parts (iii) and (iv), derivatives involving $D_F(q, p)$ are taken with respect to $q$.
\end{lemma}

The above result allows us to establish the following upper and lower bounds on the value, gradient, and Hessian of a Bregman divergence based on the maximum and minimum eigenvalues of the function's Hessian.

\begin{lemma}
\label{l:bregman-bounds}
Let $F$ be a strictly convex function defined over some domain $\XX \subseteq \RE^d$, and let $D_F$ denote the associated Bregman divergence. For each $x \in \XX$, let $\lambda_{\min}(x)$ and $\lambda_{\max}(x)$ denote the minimum and maximum eigenvalues of $\Hess F(x)$, respectively. Then, for all $p, q \in \XX$, there exist points $r_1$, $r_2$, and $r_3$ on the open line segment $\overline{p q}$ such that
\[ 
    \begin{array}{rcccl}
    \frac{1}{2} \lambda_{\min}(r_1) \|q - p\|^2
        & \leq & D_F(q, p)
        & \leq & \frac{1}{2} \lambda_{\max}(r_1) \|q - p\|^2 \\[5pt]
    \lambda_{\min}(r_2) \|q - p\|
        & \leq & \| \Gradient D_F(q,p) \|
        & \leq & \lambda_{\max}(r_3) \|q - p\| \\[5pt]
    \lambda_{\min}(q)
        & \leq & \| \Hess D_F(q, p) \|
        & \leq & \lambda_{\max}(q).
      \end{array}
\]
\end{lemma}

\begin{proof}
To establish the first inequality, we apply Taylor's theorem with the Lagrange form of the remainder to obtain
\[
    F(q) 
        ~ = ~ F(p) + \ang{\Gradient F(p), q - p} + \frac{1}{2} (q - p)^{\Transpose} \Hess F(r_1) (q - p),
\]
for some $r_1$ on the open line segment $\overline{p q}$. By substituting the above expression for $F(q)$ into the definition of $D_F(q, p)$ we obtain
\[
    D_F(q, p) 
           =   ~ F(q) - F(p) - \ang{\Gradient F(p), q - p}
           =   ~ \frac{1}{2} (q - p)^{\Transpose} \Hess F(r_1) (q - p).
\]

By basic linear algebra,
\[
    \lambda_{\min}(r_1) \|q - p\|^2
        ~ \leq ~ (q - p)^{\Transpose} \Hess F(r_1) (q - p)
        ~ \leq ~ \lambda_{\max}(r_1) \|q - p\|^2.
\]
Therefore,
\[
    \frac{\lambda_{\min}(r_1)}{2} \|q - p\|^2
        ~ \leq ~ D_F(q, p)
        ~ \leq ~ \frac{\lambda_{\max}(r_1)}{2} \|q - p\|^2,
\]
which establishes the first assertion.

For the second assertion, we recall from Lemma~\ref{l:bregman-util}(iii) that $\Gradient D_F(q, p) = \Gradient F(q) - \Gradient F(p)$. Let $v$ be any unit vector. By applying the mean value theorem to the function $\psi(t) = v^{\Transpose} \Gradient F(p + t (q - p))$ for $0 \leq t \leq 1$, there exists a point $r_2 \in \overline{p q}$ (which depends on $v$) such that $v^{\Transpose} (\Gradient F(q) - \Gradient F(p)) = v^{\Transpose} \Hess F(r_2) (q - p)$. Taking $v$ to be the unit vector in the direction of $q-p$, and applying the Cauchy-Schwarz inequality, we obtain
\begin{align*}
    \|\Gradient D_F(q,p)\| 
        & ~ = ~ \|\Gradient F(q) - \Gradient F(p)\| 
        ~ \geq ~ |v^{\Transpose} (\Gradient F(q) - \Gradient F(p))| \\
        & ~ = ~ |v^{\Transpose} \Hess F(r_2) (q - p)|
        ~ \geq ~ \lambda_{\min}(r_2) \|q - p\|.
\end{align*}
For the upper bound, we apply the same approach, but take $v$ to be the unit vector in the direction of $\Gradient F(q) - \Gradient F(p)$. There exists $r_3 \in \overline{p q}$ such that
\begin{align*}
    \|\Gradient D_F(q,p)\| 
        & ~ = ~ \|\Gradient F(q) - \Gradient F(p)\| 
        ~ = ~ |v^{\Transpose} (\Gradient F(q) - \Gradient F(p))|
        ~ = ~ |v^{\Transpose} \Hess F(r_3) (q - p)| \\
        & ~ \leq ~ \|\Hess F(r_3) (q - p)\|
        ~ \leq ~ \lambda_{\max}(r_3) \|q - p\|.
\end{align*}
This establishes the second assertion.

The final assertion follows from the fact that $\Hess D_F(q, p) = \Hess F(q)$ (Lemma~\ref{l:bregman-util}(iv)) and the definition of the spectral norm.
\end{proof}

With the help of this lemma, we can now present a proof of Lemma~\ref{l:bregman-comparisons} from Section~\ref{s:bregman-complexity}, which relates the various measures of complexity for Bregman divergences.

{\bregmancomparisons*}

\begin{proof} 
For each $x \in \XX$, let $\lambda_{\min}(x)$ and $\lambda_{\max}(x)$ denote the minimum and maximum eigenvalues of $\Hess F(x)$, respectively. We first show that for all $x \in \XX$, $2 \leq \lambda_{\min}(x)$ and $\lambda_{\max}(x) \leq 2 \mu$. We will prove only the second inequality, since the first follows by a symmetrical argument. Suppose to the contrary that there was a point $x \in \XX$ such that $\lambda_{\max}(x) > 2 \mu$. By continuity and the fact that $\XX$ is convex and open, there exists a point $q \in \XX$ distinct from $x$ such that for any $r$ on the open line segment $\overline{q x}$,
\begin{equation}
    (q - x)^{\Transpose} \Hess F(r) (q - x) 
        ~ > ~ 2 \mu \|q - x\|^2. \label{e:admissibility-1}
\end{equation}
Specifically, we may take $q$ to lie sufficiently close to $x$ along $x + v$, where $v$ is the eigenvector associated with $\lambda_{\max}(x)$. As in the proof of Lemma~\ref{l:bregman-bounds}, we apply  Taylor's theorem with the Lagrange form of the remainder to obtain
\[
    D_F(q, x) 
        ~ = ~ F(q) - F(x) - \ang{\Gradient F(x), q - x} 
        ~ = ~ \frac{1}{2} (q - x)^{\Transpose} \Hess F(r) (q - x).
\]
By Eq.~\eqref{e:admissibility-1}, we have $D_F(q, x) > \mu \|q - x\|^2$. Therefore, $D_F$ is not $\mu$-similar. This yields the desired contradiction.

Because $2 \leq \lambda_{\min}(x) \leq \lambda_{\max}(x) \leq 2 \mu$ for all $x \in \XX$, by Lemma~\ref{l:bregman-bounds}, we have
\[ 
    \|q - p\|^2
        ~ \leq ~ D_F(q, p), \qquad
    \| \Gradient D_F(q,p) \|
        ~ \leq ~ 2 \mu \|q - p\|, \quad\hbox{and}\quad
    \| \Hess D_F(q, p) \|
        ~ \leq ~  2 \mu,
\]
which imply
\[
    \| \Gradient D_F(q, p) \| \kern+2pt \|q - p\|
        ~ \leq ~ 2 \mu D_F(q, p) 
    \qquad\hbox{and}\qquad
    \| \Hess D_F(q, p) \| \kern+2pt \|q - p\|^2
        ~ \leq ~ 2 \mu D_F(q, p),
\]
which together imply that $D$ is $2 \mu$-admissible, as desired.

To prove~(ii), observe that by the Cauchy-Schwarz inequality $\ang{\Gradient D_F(q, p), q - p} \leq \| \Gradient D_F(q, p)\| \cdot \| q - p \|$, and therefore, any divergence that satisfies the condition for $\mu$-admissibility immediately satisfies the condition for directional $\mu$-admissibility.

Finally, to show~(iii), consider any points $p, q \in \XX$. Recall the facts regarding the Bregman divergence presented in Lemma~\ref{l:bregman-util}. By combining observations~(i) and~(iii) from that lemma, we have $D_F(q, p) + D_F(p, q) = \ang{\Gradient D_F(q, p), q - p}$. Observe that if $D$ is directionally $(1 + \mu)$-admissible, then
\[
    D_F(q, p) + D_F(p, q)
        ~  =   ~ \ang{\Gradient D_F(q, p), q - p} 
        ~ \leq ~ (1 + \mu) D_F(q, p),
\]
which implies that $D_F(p, q) \leq \mu (D_F(q, p))$, and hence $D$ is $\mu$-asymmetric. Conversely, if $D$ is $\mu$-asymmetric, then
\[
    \ang{\Gradient D_F(q, p), q - p}
        ~  =   ~ D_F(q, p) + D_F(p, q) \\
        ~ \leq ~ D_F(q, p) + \mu \kern+1pt D_F(q, p) 
        ~  =   ~ (1 + \mu) D_F(q, p),
\]
implying that $D_F$ is directionally $(1 + \mu)$-admissible. (Recall that directional admissibility requires only that the gradient condition be satisfied.)
\end{proof}

Next, we provide a proof of Lemma~\ref{l:bregman-inner} from Section~\ref{s:bregman-data-struct}.

{\bregmaninner*}

\begin{proof} 
Without loss of generality, we may assume that $D(q, p) \geq D(q, p')$. By adding $D(p, p')$ to the left side of Lemma~\ref{l:bregman-util}(ii) and rearranging terms, we have
\begin{align*}
    D(q, p) - D(q, p')
        & \leq  ~ (D(q, p) - D(q, p')) + D(p, p') \\
        &   =   ~ (D(p', p) + \ang{\Gradient F(p') - \Gradient F(p), q - p'}) +  D(p, p') \\
        &   =   ~ \ang{\Gradient F(p') - \Gradient F(p), q - p'} +  (D(p', p) + D(p, p')).
\end{align*}
By Lemma~\ref{l:bregman-util}(i) we have
\begin{align*}
    D(q, p) - D(q, p')
        & ~ \leq  ~ \ang{\Gradient F(p') - \Gradient F(p), q - p'} +  \ang{\Gradient F(p') - \Gradient F(p), p' - p} \\
        & ~  =    ~ \ang{\Gradient F(p') - \Gradient F(p), q - p}.
\end{align*}

Let $v$ be any vector. Applying the mean value theorem to the function $\psi(t) = v^{\Transpose} \Gradient F(p + t (p' - p))$ for $0 \leq t \leq 1$, implies that there exists a point $r \in \overline{p p'}$ (which depends on $v$) such that $v^{\Transpose} (\Gradient F(p') - \Gradient F(p)) = v^{\Transpose} \Hess F(r) (p' - p)$. Letting $v = q-p$, and applying the Cauchy-Schwarz inequality, we obtain
\[
    D(q, p) - D(q, p')
        ~ \leq  ~ \ang{\Hess F(r) (p' - p), q - p} 
        ~ \leq  ~ \| \Hess F(r) \| \kern+1pt \|p' - p\| \kern+1pt \|q - p\|.
\]
By Lemma~\ref{l:bregman-util}(iv) and $\tau$-admissibility, $\|\Hess F(r)\| = \|\Hess D(r, q)\| \leq \tau D(r, q) / \|r - q\|^2$, which implies
\begin{equation}
    D(q, p) - D(q, p')
        ~ \leq  ~ \frac{\tau D(r, q)}{\|r - q\|^2} \kern+1pt \|p' - p\| \kern+1pt \|q - p\|. \label{e:bregman-inner-1}
\end{equation}
Since $r$ lies on the segment between $p'$ and $p$, it follows that $r \in B_w$. Letting $\delta = \diam(B_w)$, we have $\max(\|p' - p\|, \|r - p\|) \leq \delta$ and $\|r - q\| \geq \beta \delta$. By the triangle inequality, $\|q - p\| \leq \|q - r\| + \|r - p\|$. Therefore,
\[
    \frac{\|q - p\|}{\|r - q\|} 
        ~ \leq  ~ \frac{\|q - r\| + \|r - p\|}{\|r - q\|}
        ~   =   ~ 1 + \frac{\|r - p\|}{\|r - q\|}
        ~ \leq  ~  1 + \frac{1}{\beta},
\]
and since clearly $\beta \geq 1$, 
\begin{equation}
    \frac{\|p' - p\| \kern+1pt \|q - p\|}{\|r - q\|^2}
        ~ \leq  ~ \frac{1}{\beta} \left(1 + \frac{1}{\beta}\right)
        ~ \leq  ~ \frac{2}{\beta}. \label{e:bregman-inner-2}
\end{equation}
We would like to express the right-hand side of Eq.~\eqref{e:bregman-inner-1} in terms of $p$ rather than $r$. By the $\tau$-admissibility of $D$ and the fact that $r, p \in B_w$, we can apply Lemma~\ref{l:admiss-util}(i) (with the distance function $f_q(\cdot) = D(\cdot, q)$ and $\kappa = \beta/\tau$) to obtain $D(r, q) \leq D(p, q) / (1 - \tau/\beta)$. Combining Eq.~\eqref{e:bregman-inner-2} with this, we obtain
\[
    D(q, p) - D(q, p')
        ~ \leq  ~  \frac{2 \tau}{\beta} D(r, q)
        ~ \leq  ~  \frac{2 \tau}{\beta (1 - \tau/\beta)} D(p, q).
\]
In Lemma~\ref{l:bregman-comparisons}(iii) we showed that any $(1 + \mu)$-admissible Bregman divergence is $\mu$-asymmetric, and by setting $\mu = \tau - 1$ it follows that $D(p, q) \le (\tau - 1) D(q, p)$. Putting this all together, we obtain
\[
    D(q, p) - D(q, p')
        ~ \leq  ~  \frac{2 \tau (\tau - 1)}{\beta (1 - \tau/\beta)} D(q, p).
\]
All that remains is to set $\beta$ sufficiently large to obtain the desired result. Since $\tau \geq 1$ and $\eps \leq 1$, it is easily verified that setting $\beta = 4 \tau^2/\eps$ suffices to produce the desired conclusion.
\end{proof}


\pdfbookmark[1]{References}{s:ref}

\end{document}